\documentclass[12pt]{iopart}

\usepackage{graphicx}
\usepackage[normalem]{ulem}
\usepackage{array,amssymb,amsthm}
\usepackage{dcolumn}
\usepackage{bm}
\usepackage{hyperref}
\usepackage{xcolor}
\usepackage[noend]{algpseudocode}
\usepackage{amsfonts}
\usepackage[utf8x]{inputenc}
\usepackage{algorithm2e}
\usepackage{pseudocode}
\usepackage{multirow}
\usepackage{caption}
\usepackage{subcaption}
\usepackage{makecell}

\expandafter\let\csname equation*\endcsname\relax

\expandafter\let\csname endequation*\endcsname\relax

\usepackage{amsmath}

\newcommand{\arc}{\mathrm{ARC}}
\newcommand{\ket}[1]{|#1\rangle}

\newtheorem{lemma}{Lemma}
 
\begin{document}

\title[]{Entanglement Distribution in Multi-Platform Buffered-Router-Assisted Frequency-Multiplexed Automated Repeater Chains}

\author{Mohsen Falamarzi Askarani$^{1}$\footnote{Present address: Xanadu Quantum Technologies Inc., 777 Bay Street, Toronto, ON M5G 2C8, Canada}, Kaushik Chakraborty$^{1,2}$, Gustavo Castro do Amaral$^{1,3}$\footnote{Present address: TNO, Stieltjesweg 1, 2628 CK Delft, The Netherlands}}

\address{$^1$QuTech, Delft University of Technology, 2628 CJ Delft, The Netherlands}
\address{$^2$Informatics Forum, The University of Edinburgh, 10 Crichton St, Newington, Edinburgh EH8 9AB, United Kingdom}
\address{$^3$Center for Telecommunication Studies, Pontifical Catholic University of Rio de Janeiro, Rio de Janeiro, 22451-900, Brazil}
\ead{amaral@puc-rio.br}

\begin{abstract}
Realization of a globe-spanning quantum network is a current worldwide goal, where near and long term implementations will benefit from connectivity between platforms optimized for specific tasks. Towards this goal, a quantum network architecture is herewith proposed whereby quantum processing devices based on NV$^-$ colour centers act as quantum routers and, between which, long-distance entanglement distribution is enabled by spectrally-multiplexed quantum repeaters based on absorptive quantum memories in rare-earth ion-doped crystals and imperfect entangled photon-pair sources. The inclusion of a quantum buffer structure between repeaters and routers is shown to, albeit the increased complexity, improve the achievable entanglement distribution rates in the network. Although the expected rate and fidelity results are presented for a simple linear network (point-to-point), complex topologies are compatible with the proposed architecture through the inclusion of an extra layer of temporal multiplexing in the quantum router's operation. Figures of merit are extracted based on parameters found in the literature for near-term scenarios and attest the availability of the proposed buffered-router-assisted frequency-multiplexed automated repeater chain network.
\end{abstract}

\section{Introduction}

    A large-scale quantum communication network is essential to connect multiple quantum information processing nodes as well as to implement interesting applications like quantum key distribution, clock synchronisation and others \cite{wehner2018Qinternet}. The crucial communication task of this quantum internet is to distribute entanglement among any set of its nodes. The constituents of such network can, thus, be categorized according to their functionality: those responsible for processing quantum information, effectively implementing the protocols enabled by the quantum network; and those responsible for relaying quantum information across the network, providing connectivity between processing nodes. Due to the limit imposed by channel losses, quantum repeaters are necessary in order to achieve long-distance entanglement distribution, a technology that falls under the auspices of the second category.
	
	As had already been experienced during the deployment of the current global classical internet, managing the information flow across the network is not simple, and information processing structures are also required in the second category. For example, recently Rabbie \textit{et al} \cite{rabbie2020designing} showed that, in order to have a quantum network with high connectivity, some of the repeater nodes should be connected to many other repeater nodes. This type of repeater node should, thus, have extra information processing power in order to serve the multiple requests. Furthermore, depending on the repeater technology, the time between a request being issued in the network and an EPR-pair being available can reach high values. To avoid this issue, the authors in \cite{rabbie2020designing} proposed a quantum network with pre-shared entanglement. In these types of networks the nodes pre-establish the EPR-pairs and store them inside their quantum memories; these pre-shared EPR-pairs can be used for the rapid distribution of entanglement between two long-distance nodes in the network. However, the performance of such networks depend highly on the storage time as well as the efficiency of the quantum memories and the entanglement generation rate.

	Here, we focus on the connectivity category of the quantum network while considering the inclusion of quantum processing devices to achieve better performance compared to the proposal set-forth in \cite{rabbie2020designing}. Moreover, pre-generation of EPR pairs is currently not considered to alleviate the requirements of the quantum memories' storage time; by taking advantage of heavy multiplexing capacity, however, the throughput of the network can be increased. Consider, thus, the proposed network architecture, pictorially depicted in Fig. \ref{fig:networkArch}, where three distinct structures are identified: the quantum router (QR); the automated repeater chain (ARC) \cite{sangouard2011QR}; and the quantum buffer (B).
	
	\begin{figure}[ht]
		\centering
		\includegraphics[width=0.7\linewidth]{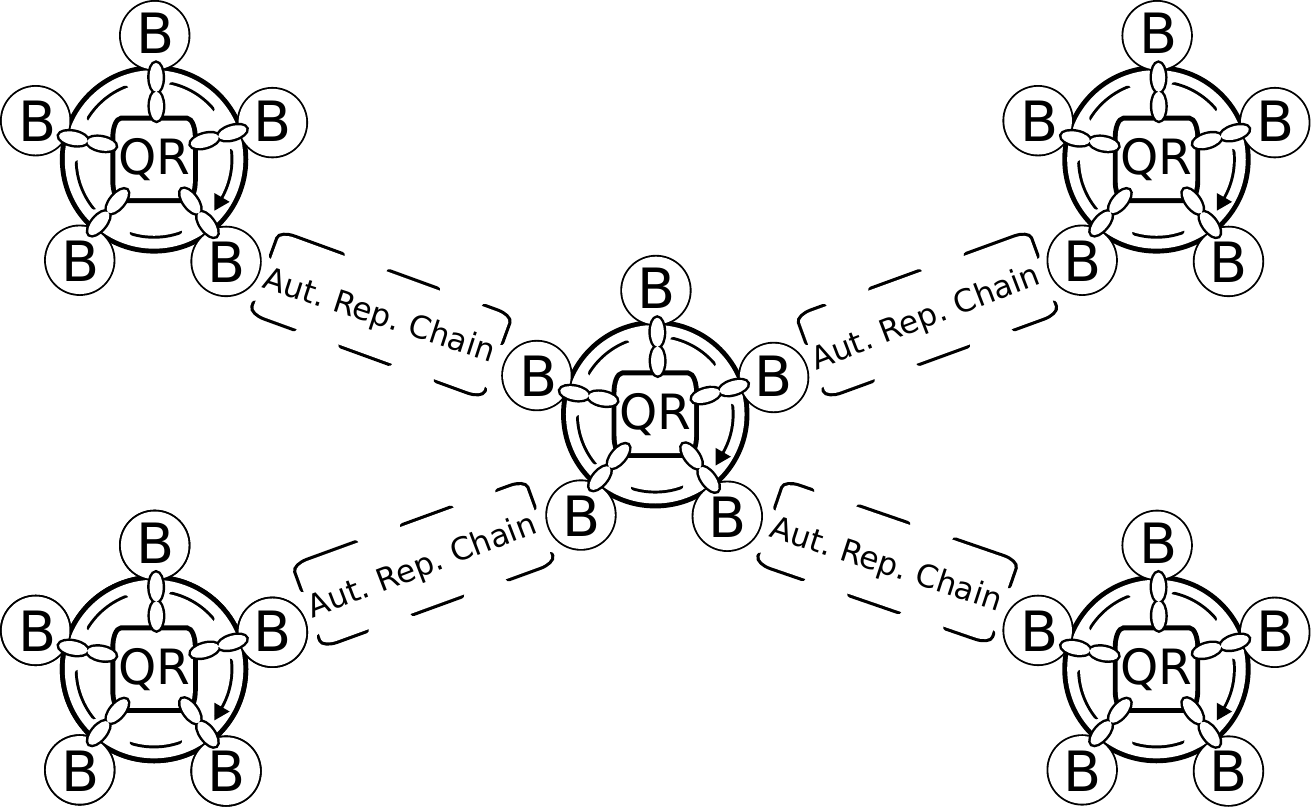}
		\caption{Outline of the proposed network architecture. The routers (QR) are directly connected to buffers (B), which, in turn, are connected through automated repeater chains (ARC).}
		\label{fig:networkArch}
	\end{figure}
	
    \begin{table}[!htb]
    \centering
    \caption{Task, Expected Characteristics, and Candidate Technologies} \label{table:requirements}
    \setlength\tabcolsep{5pt}
    \footnotesize\centering
    \smallskip 
    \begin{tabular}{c|c|c|c|c}
    Role & Task & Candidate Technology & Characteristics & Protocol\\
    \hline
    \hline
      \multirow{4}{5em}{Automated\\Repeater\\Chain} & \multirow{4}{10em}{Multiplexed and Heralded\\Entanglement Distribution} & \multirow{4}{10em}{SPDC-based EPP Sources and Broadband REI-based AFC QMs} & \multirow{4}{8em}{Inhomogeneously-broadened absorption profile} & \multirow{4}{5em}{Spectral Multiplexing}\\
      & & & &\\
      & & & &\\
      & & & &\\
      \hline
      \multirow{4}{5em}{Quantum\\ Buffer} & \multirow{4}{10em}{On-demand Quantum Storage and Quantum State Transfer} & \multirow{4}{10em}{Rubidium-based GEM QMs} & \multirow{4}{8em}{Over ms Storage Time and Unit Efficiency} & \multirow{4}{5em}{Temporal Multiplexing}\\
      & & & &\\
      & & & &\\
      & & & &\\
      \hline
      \multirow{4}{5em}{Quantum\\ Router} & \multirow{4}{10em}{Long Storage Time QM and Deterministic Quantum Information Processing} & \multirow{4}{10em}{NV$^{-}$ Color Center and $^{13}$C Nuclear Spins QMs} & \multirow{4}{8em}{QI Processing and Approaching-minute Coherence Time}  & \multirow{4}{5em}{Temporal Multiplexing}\\
      & & & &\\
      & & & &\\
      & & & &\\
    \hline
    \hline
    \end{tabular}
    \vspace{-4pt}
    \begin{flushleft}
    \scriptsize{$^{\ast}$ SPDC: Spontaneous Parametric Down Conversion, EPP: Entangled Photon-Pair, REI: Rare-Earth Ion, AFC: Atomic Frequency Comb, QM: Quantum Memory, NV: Nitrogen Vacancy, $^{13}C:$ Carbon 13, GEM: Gradient Echo Memory.}
    \end{flushleft}
    \end{table}
	
	The QR is a sophisticated unit with quantum information processing capacity that could be connected to multiple automated repeater chains in a complex network topology. ARCs, on the other hand, exhibit close to no information processing functionality but can deliver multiplexed and heralded entanglement distribution. The quantum buffer, in turn, provides an interface between these two structures. Despite the introduction of an extra layer of complexity, the buffer diminishes the idle time of the routers by storing the successfully heralded entangled pairs within the ARCs and making this information available when the routers are ready to process it. This is an especially sensitive task due to the unavoidably probabilistic nature of the entanglement swapping operations that connects elementary links composing the ARCs  \cite{BSMeff}.
	
	Over the last decade, several atomic and spin systems have been introduced as candidates for the underlying physical system of the backbone of the quantum internet. Currently, however, the different features exhibited by each candidate make it difficult to determine which will indeed be employed. By subdividing the tasks of routers, buffers, and the ARCs, the constraints on the underlying physical systems that implement each of these roles can be positively leveraged, i.e., by identifying tasks that can be facilitated by the available features of each system, thus, allowing for the development of a heterogeneous network architecture with enhanced performance.
	
	Table \ref{table:requirements} discriminates sought-after characteristics and possible candidates for each of the roles of the proposed network. Even though other candidates could be considered, this comparison is not within the scope of this work. Here, we perform a comprehensive study of the parameter regime where multi-platform quantum repeaters composed of negatively-charged Nitrogen vacancy (NV$^{-}$) colour centers and rare-earth ion-doped crystals (REICs) outperform networks built with a single constituent. We go on to show that interfacing automated repeater chains and quantum routers by means of a quantum buffer structure allows even higher entanglement distribution rates to be achieved. 
	
	Throughout this document, emphasis will be given to specific protocols and technologies that will serve as the building blocks of the quantum network, i.e., ARCs, routers, and buffers, as presented in Table \ref{table:requirements}. Other than allowing for realistic figures of merit to be extracted in order for estimation of the proposed network's capacity (in the near- and long-terms), choices are based on the compatibility of each of these physical systems to the assigned roles. This allows stratifying the network in layers that leverage complexity and multiplexing: in the lower layer, quantum memories with pre-set storage time and high level of spectral multiplexing compose the ARCs; in the intermediate layer, on-demand quantum memories are assigned to the buffers; and, finally, on the higher layer, quantum memories capable of quantum information processing implement the routers. Naturally, connectivity between the layers is essential and is discussed in the following \textcolor{black}{section. The entanglement generation protocol considered for the proposed network is introduced in Section 3. Section 4 considers perfect but finite lifetime components to derive lemmas for the network entanglement distribution rate analysis and Section 5 takes decoherence in the components into account by considering effective fidelities of noisy entangled states to estimate fidelity of the final distributed states. Results discussion and concluding remarks are presented in Sections 6 and 7, respectively. }. Before that, we present, in the next subsection, the technical details of the platforms considered in the proposed network in order to highlight the overlap with the expected characteristics depicted in Table \ref{table:requirements}.
	
	\subsection{Technical Details}
     
    \textit{Quantum Memories (QMs)}: Absorptive spectrally-multiplexed quantum memories exhibiting pre-programmed storage times and large storage bandwidth have been demonstrated in REICs \cite{lvovsky2009optical,sinclair2014spectral}. REICs offer long optical coherence lifetime, long-lived spin level lifetime, and wide absorption spectrum; combination of all in conjunction with the atomic frequency comb (AFC) memory protocol, whose multi-mode capacity is exceptionally independent of the number of available ions, enables creation of spectrally-multiplexed quantum memories. Here, Tm$^{3+}$:Y$_{3}$Ga$_{5}$O$_{12}$ (Tm:YGG) crystal \cite{thiel2014Tm:YGG,MFA2021Tm:YGG} is considered for the role of QM, as it exhibits sufficiently long optical coherence lifetime -- an upper-bound for the pre-programmed storage time --, 56 GHz-wide absorption window centred at 795nm wavelength -- which ultimately determines the maximum number of receptive spectral modes --, and spin level lifetime as long as tens of seconds -- which guarantees the persistency of the AFC memory. The AFC memory preparation is performed via so-called frequency-selective optical pumping, which takes nearly one second in Tm:YGG crystal \cite{thiel2014Tm:YGG,MFA2021Tm:YGG} and it is well below the lifetime of the long-lived shelving level (the ground-state Zeeman level), which is much longer than the 1.3~ms excited level lifetime -- a stringent requirement for efficient optical pumping. A few tens of milliseconds of waiting time are necessary after optical pumping to avoid detrimental spontaneous processes during the operational time. The total experimental time is given by the persistency of the AFC memory, which is upper-bounded by the longevity of the ground-state Zeeman upper level, guaranteeing good compromise between preparation and experimental time duration.
     
    \textit{EPPSs}: In the ARCs considered here, entangled-photon-pairs at distinct spectral modes are generated by means of spontaneous parametric down-conversion (SPDC) in a second-order non-linear crystal pumped by a strong input signal, which determines the repetition rate of the EPPS. While the output spectral modes at 795nm wavelength are matched with the storage bandwidth of the QMs \cite{Entan2QMs}, the remaining photon of the pair exhibits a wavelength of 1535nm, compatible with low-loss optical fiber transmission.
     
    \textit{Buffers}: Quantum buffers are realized using rubidium vapours, and operated according to the gradient echo memory (GEM) protocol, that exhibit an on-demand storage time of several milliseconds, in-principle unit retrieval efficiency, and narrow storage bandwidth, whose central wavelength can be chosen to match the wavelength of the absorptive quantum memories, in this case, 795nm. The $^{87}$Rb buffer is prepared through a few different phases, starting with a magneto-optical trap (MOT) loading phase, that is implemented to obtain sufficient interactable ions, followed by the application of a magnetic field gradient required for the GEM protocol, which collectively takes a few hundreds of millisecond. A short optical pumping phase follows, whereby the desired Zeeman sub-level is populated with available ions, and, finally, an idle phase with waiting time of approximately one milliseconds prior to the qubit storage and retrieval phase \cite{RbGEM}. Again for the buffer, as for the QMs, preparation and experimental time duration allows efficient usage.
     
    \textit{Routers}: The negatively-charged Nitrogen Vacancy center-- a defect in diamond-- has been extensively explored for quantum information applications as it can be addressed with both optical and microwave signals. In addition to the electron spin (S=1) of the NV$^{-}$ center, it carries a nuclear spin of I=1 associated with the nitrogen atom whose $^{14}$N is the most common isotope. The energy levels of the electron and nuclear spins can be manipulated through Zeeman effect---by means of external magnetic field---and hyperfine interaction, which mediates the coupling between the electron and nuclear spins; such nuclear spins can be used for quantum state storage thanks to their very long coherence time. Due to the same reason, surrounding Carbon nuclear spins, in particular the $^{13}$C isotope, can be also employed as storage medium. As thoroughly detailed in refs \cite{nemoto2014cavitymapping, nemoto2016photonic}, an incoming photonic quantum state can interact with the NV$^{-}$ center and, in turn, with the $^{14}$N or $^{13}$C nuclear spin through an NV-embedded optical cavity system. Thanks to very long coherence times of $^{13}$C nuclear spins  surrounding the NV$^{-}$ electron spin and the hyperfine coupling \cite{abobeih2018one}, the quantum state is transferred and stored into one of these nuclear spins. After the transfer, the electron spin is again free to accept another photon.

    \section{Building blocks and connectivity of the proposed network}\label{sect:II}
    
    \begin{figure}
    \captionsetup[subfigure]{justification=centerlast,singlelinecheck=off}
    \centering
    \begin{subfigure}[b]{0.81\linewidth}
       \includegraphics[width=1\linewidth]{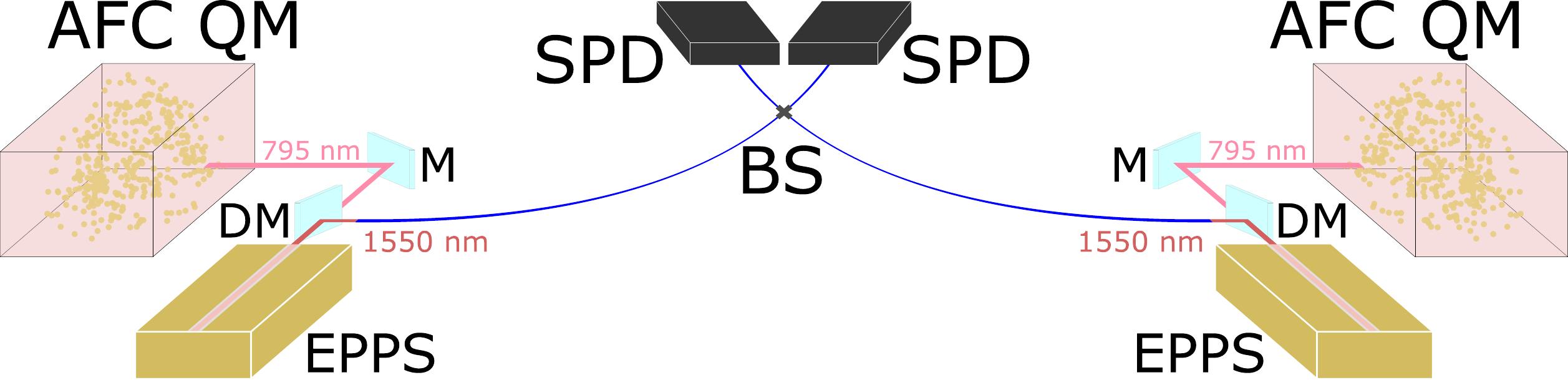}
       \caption{Elementary link connectivity. Successful projection of the quantum states encoded into two 1550 nm photons from different sources onto the Bell-basis casts the states stored in the QMs into an entangled state. DM: Dichroic mirror; M: mirror; SPD: single-photon detector; BS: optical beamsplitter.}
       \label{fig:ElementaryLink} 
    \end{subfigure}
    ~
    
    \begin{subfigure}[b]{0.81\linewidth}
       \includegraphics[width=1\linewidth]{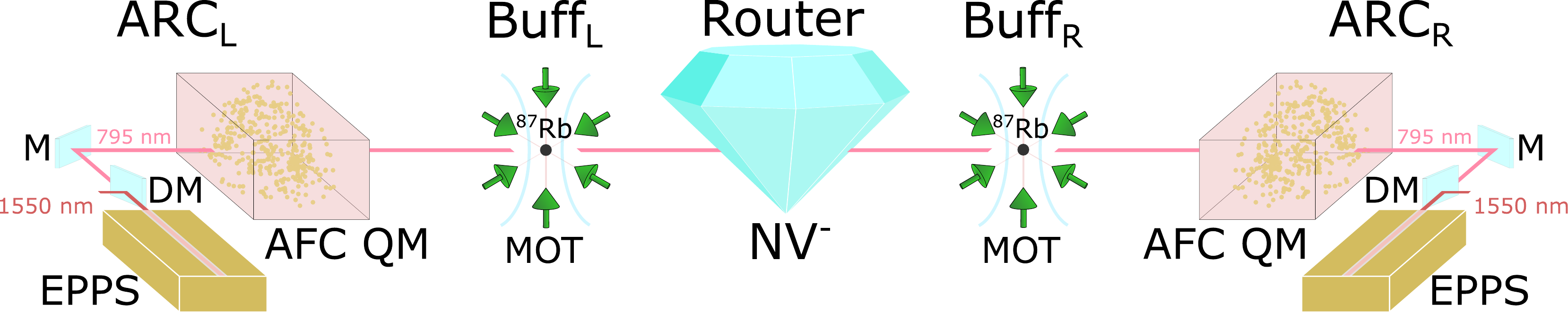}
       \caption{Connectivity between two adjacent ARCs mediated by a QR assisted by buffers. MOT: Magneto-Optical trap.}
       \label{fig:RouterARC_connect}
    \end{subfigure}
    ~
    
    \begin{subfigure}[b]{0.81\linewidth}
       \includegraphics[width=1\linewidth]{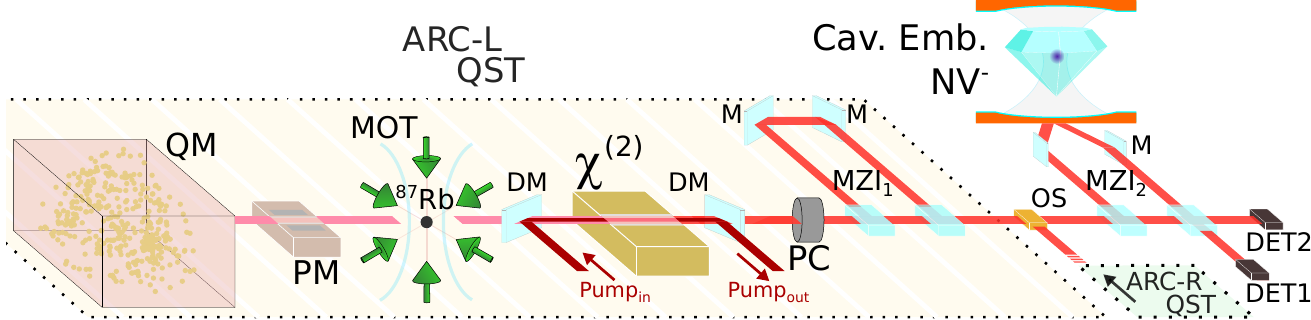}
       \caption{Detailed quantum state transfer (QST) from QM, through buffer, into the internal memory of the QR. PM: Electro-optical phase-modulator; MZI$_i$: Mach-Zender interferometer; OS: optical switch; $\chi^{\text{(2)}}$: non-linear crystal.}
       \label{fig:quantStateTransfer}
    \end{subfigure}
    \caption{Connectivity between elements of the proposed architecture.}
    \label{fig:experimental setup}
    \end{figure}
    
    Although the proposed network could accept a complex topology, whereby routers would be connected to more than two ARCs \cite{rabbie2020designing}, we use, throughout the analysis, the basic example of an one-dimensional repeater link for simplicity. In this particular topology, two so-called end nodes, with advanced quantum computation capacity, attempt to share entanglement (in the form of an EPR-pair) through a point-to-point quantum channel. In our proposal, the quantum channel is composed of ARCs connected to routers by means of quantum buffers. The ARCs, in turn, are made up of concatenated elementary links, a possible realization of which is depicted in Fig. \ref{fig:ElementaryLink}.
    
    In this realization (AFC-based ARC), entanglement is distributed between two quantum memories within an elementary link by means of entangled photon pair sources and linear-optic Bell state projection measurements (BSM) that comprise the entanglement swapping protocol; presently, the analysis considers a double-click protocol, whereby quantum information is encoded into the temporal degree-of-freedom of flying photons (time-bin encoding \cite{Gisin2002Time-binQB}). After a successful entanglement swapping operation is heralded by two consecutive elementary links, interconnection through another round of entanglement swapping is achieved by retrieving the quantum states stored in the quantum memories and performing yet another BSM \cite{sangouard2011QR}. The storage time of the quantum memories in all of the repeaters inside the ARCs are set to be $t^{\text{AFC}}$, which must be smaller than the decoherence time of the memory.  

    In case successful entanglement swapping within and between all the elementary links that constitute a particular ARC is heralded, the quantum memories at either edges of the chain will be, ideally, storing individual parties of a bipartite maximally entangled state. Unfortunately, the particular case described above scarcely happens as the number of concatenated elementary links increases: the very low individual success probabilities are combined to produce an extremely low overall success probability. To counteract this harsh scaling \cite{sangouard2011QR}, the multiplexing feature enabled by the ARC allows for enhanced entanglement distribution rate between the quantum memories within elementary links. In the present context, multiplexing refers to spectral multiplexing, where entangled photon pairs are created occupying distinct spectral modes and can, therefore, be multiplexed into the same temporal and spatial mode; this requires the entangled photon pair sources to be able to generate such spectrally-multiplexed entangled photonic states and the quantum memories to be sufficiently broadband such that they can store all the generated spectral modes \cite{sinclair2014spectral}.

    Let us discuss the situation where two ARCs, connected to either side -- left (ARC-L) and right (ARC-R), according to Fig. \ref{fig:RouterARC_connect}) -- of a quantum router, have been able to successfully herald entanglement distribution between their respective edge quantum memories. The goal for the quantum router is to load the quantum states of both QMs, between which a CNOT-based BSM operation can be performed within the router, so that, finally, the far-end QMs of both ARCs become entangled following the entanglement swapping protocol. A problem arises at this point since the quantum router, implemented by an NV$^{-}$ colour center in this proposal, can only load the quantum state from one of the ARC's at a time. Due to the limited pre-programmed storage time of the QMs, which is at least one order of magnitude smaller than the time necessary for the quantum state to be loaded onto the quantum router, the quantum state will be lost. Thus, it is necessary to introduce another device, namely a quantum buffer, able to interface the quantum state transfer, effectively mediating the router-loading steps.

    Such a quantum buffer must hold on to the quantum states produced at the edges of the ARCs whenever successful entanglement swapping across the latter is heralded; thus, a long-lived quantum memory is necessary; on-demand storage time is also necessary so that the states can be accessed whenever the router is available. It is important to note that, at the edges of the ARCs, the time-bin photonic qubits are no longer spectrally-multiplexed, since only the successful spectral mode has been recovered after a chain of feed-forward spectral mode-mapping (FFSMM) operations \cite{FeedforwardMulti}. Therefore, the buffer has no constraints in terms of spectral broadness and, thus, narrow-linewidth, long-lived, on-demand quantum memories can be employed for this task, such as presented in Table \ref{table:requirements}. Provided that the storage time of the buffer is sufficiently long, the quantum states can be mapped, one after the other, onto the quantum router so that inter-connectivity across the whole network is achieved.

    A critical aspect of the proposed network architecture is the analysis of the quantum state transfer from the QM on the edge of the ARCs, through the buffer, and into the router's internal memory. The proposed setup is depicted in Fig. \ref{fig:quantStateTransfer} and works as follows. After a pre-programmed re-emission from the QM, the spectrally multiplexed time-bin photonic qubits undergo a frequency-shifting and filtering process based on the information from the remote BSM, which constitutes the FFSMM; at the output of the FFSMM, only the successful mode is present. The overall efficiency of this process is denoted by $\eta^{\text{shift}}$. The successful mode is directed to the buffer, which is prepared to absorb, store and retrieve the time-bin photonic qubit in an on-demand fashion that is triggered based on the availability of the QR; the operational optical wavelength of QM and buffer is the same, requiring no frequency conversion at this stage. The buffer storage efficiency and maximum storage time are represented by $\eta^{\text{BUFF}}$ and $t^{\text{BUFF}}_{\text{spin}}$, respectively.

    After retrieval from the buffer, the operational wavelength must be converted from 795nm to 637nm so that it is compatible with the transition frequency of the NV. To do so, a quantum frequency conversion apparatus is considered, where a strong pump at $\lambda_{\text{pump}}=3.4~\mu$m interacts with the input qubit at $\lambda_{\text{in}} = 795$~nm in a non-linear crystal, resulting in up-conversion to $\lambda_{\text{out}}=637$~nm; this process is characterized by a conversion efficiency $\eta^{\text{QFC}}_{637}$. In order for the time-bin photonic qubit to be mapped onto the electron spin state of the NV using the protocol set forth in \cite{nemoto2014cavitymapping}, time-bin to polarization mapping is necessary -- an equivalent, but reversed, mapping has been demonstrated in \cite{Pan-PolarizBinToTimebin}. This is enforced by a combination of a fast Pockels cell -- that switches the polarization state of the early and late time-bins such that $\ket{e}\longrightarrow\ket{e,H}$ and $\ket{\ell}\longrightarrow\ket{\ell,V}$ -- and an unbalanced Mach-Zender interferometer (MZI$_1$) -- which erases the time information and outputs a polarization photonic qubit with efficiency $\eta^{\text{POL}}$ \cite{Pan-PolarizBinToTimebin}. 
    
    The NV-embedded optical cavity is placed inside a Mach-Zender-like-interferometer (MZI$_2$) with polarizing beam-splitters (PBS). A note is made, here, that, due to the similarity between MZIs 1 and 2, a single interferometer may be sufficient to perform the mapping of a time-bin photonic qubit directly to the NV's electronic spin; this possibility is not currently considered for simplicity. Before MZI$_2$, a photonic switch is connected, allowing the output of different ARCs (-L and -R) to be directed to the router in a synchronous fashion. Mapping onto the electron spin is performed conditioned on a successful detection pattern at the detectors connected to the output of MZI$_2$, a process characterized by its mapping efficiency $\eta^{\text{MAP}}$. Finally, a mapping from the electron spin onto a long-lived $^{13}$C internal memory takes place, with duration $t^{^{13}\text{C}}$ and efficiency $\eta^{^{13}\text{C}}$. With the electron available, the state retrieved from the remaining buffer can be loaded by acting on the optical switch and following the same procedure; connectivity between two ARCs through a QR can, finally, be achieved after a CNOT-based BSM between two $^{13}$C memories. We abstract the total efficiency of the entire process, i.e., from storage and retrieval within the buffer, all the way to storage inside the $^{13}$C memory (from left to right according to Fig. \ref{fig:quantStateTransfer}), with the parameter $\eta^{\text{QR}}$:
    \begin{equation}
        \label{eq:pbuf}
        \eta^{\text{QR}} = \left(\eta^{\text{BUFF}}\eta^{\text{QFC}}_{637}\eta^{\text{POL}}\eta^{\text{MAP}}\eta^{^{13}\text{C}}\right).
    \end{equation}
    The FFSMM operation, even though depicted in Fig. \ref{fig:quantStateTransfer}, is considered a part of the ARC operation and, as such, is not included into the quantum state transfer between the ARC and the QR.
    
    A block diagram representing the proposed connectivity between two quantum routers mediated by an AFC-based ARC with $n=3$ elementary links is depicted in Fig. \ref{fig:AFC-basedARC}. In the green-highlighted block, the probability of entanglement generation within a single elementary link is abstracted by $P^{\text{AFC}}_{\text{gen}}$, with $S$ representing EPPSs, $M$ representing QMs, and LO-BSM the linear-optic-based BSMs; $P^{\text{AFC}}_{\text{gen}}$, will be formally defined in Section \ref{sect:IV}. At the interface with the QR, $T$ represents the quantum state transfer detailed in Fig. \ref{fig:quantStateTransfer} and is broken down into $F$ -- the FFSMM operation --, $B$ -- the storage into the buffer--, $\nu$ -- the QFC step --, and $m$ -- the time-bin to polarization photonic qubit mapping.
    
    \begin{figure}[h]
	\centering
	\includegraphics[width=0.8\linewidth]{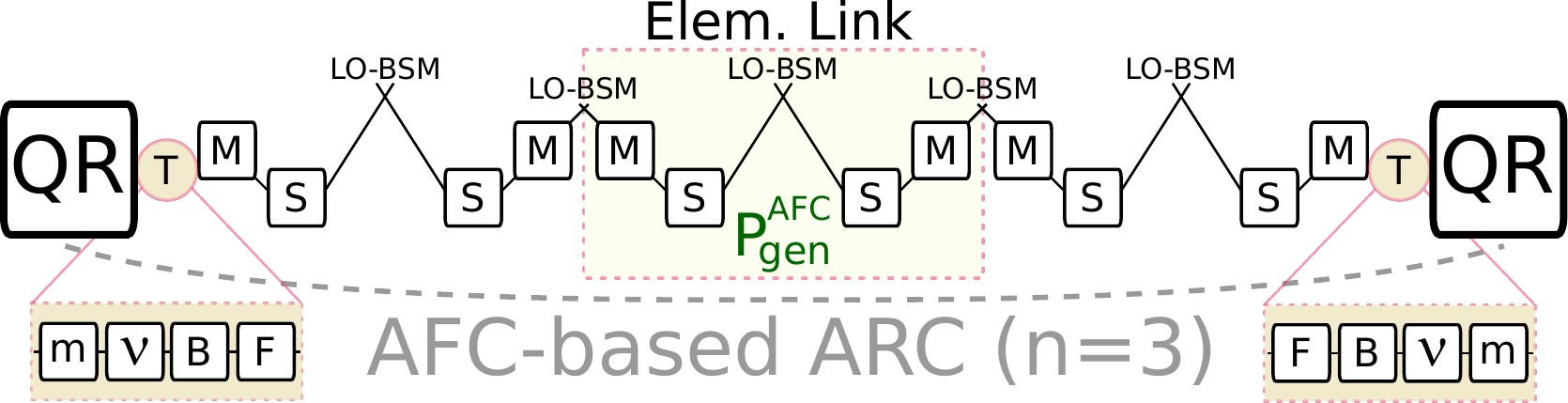}
	\caption{Pictorial representation of the proposed architecture connecting two quantum routers through an AFC-based ARC composed of $n=3$ concatenated elementary links.}
	\label{fig:AFC-basedARC}
    \end{figure}
    
    The parameters introduced here, as well as other ones relevant to the entanglement distribution rate calculation, are summarized in Table \ref{tab:params1} along with values for near- and long-term implementations as well as idealized values. All near-term parameters were extracted directly from recent experimental works and under specific experimental conditions, i.e., preparation schemes and protocols. Those values of long-term parameters were arbitrarily chosen considering the prospects of the mentioned technologies in the next 10 to 20 years. The ideal parameters take into account the absolute theoretical limits of each platform and, even though will not be included in the simulation results, allow gauging the status of the proposed network for near- and long-term parameters.
    
    \begin{table}[ht]
    \centering
    \renewcommand{\arraystretch}{1.1}
    \begin{tabular}{c | c || c | c | c}
    Symbol & Description & Near-Term & Long-Term & Ideal\\
    \hline
    \hline
    \multicolumn{5}{c}{NV$^{-}$}\\
    \hline
    $\eta^{\text{NV}}$ &
    \centering NV$^{-}$ Emission Efficiency &
    5\% \cite{NVeffCavity} & 40\% & 100\%\\
    $t^{\text{NV}}$ &
    \centering NV Storage Time &
    1s \cite{abobeih2018one} &
    10s &
    20s\\
    $t^{^{13}\text{C}}$ &
    $e^{-}\!\rightarrow\!^{13}\!C$ Swap Time &
    500$\mu$s \cite{C13gates} &
    100$\mu$s &
    10$\mu$s\\
    $\eta^{^{13}\text{C}}$ &
    $e^{-}\!\rightarrow\!^{13}\!C$ Efficiency &
    90\% \cite{C13gates} &
    99\% &
    99.9\%\\
    $t^{\text{CNOT}}$ &
    CNOT Time &
    500$\mu$s \cite{C13gates} &
    100$\mu$s &
    10$\mu$s\\
    $\eta^{\text{QFC}}_{1588}$ &
    $\lambda_{\text{NV}^{-}}\!\rightarrow \lambda_{\text{Telecom}}$ QFC Efficiency &
    43\% \cite{Pan-PolarizBinToTimebin} &
    70\% &
    99\%\\
    $\Gamma_t$ &
    Temporal Modes &
    27 \cite{C13gates} &
    100&
    1000\\
    \hline
    \multicolumn{5}{c}{AFC and EPPS}\\
    \hline
    $\eta^{\text{EPPS}}$ &
    EPPS Efficiency &
    10\% \cite{Entan2QMs} &
    10\% &
    10\%\\
    $\eta^{\text{AFC}}$ &
    AFC Memory Efficiency &
    40\% \cite{AFCcavityeff} &
    75\% &
    99\%\\
    $t^{\text{AFC}}$ &
    \centering AFC Storage Time &
    100$\mu$s \cite{thiel2014Tm:YGG,MFA2021Tm:YGG} &
    300$\mu$s &
    500$\mu$s\\
    $R^{\text{EPPS}}$ &
    EPPS Rate &
    10$^8$Hz \cite{Entan2QMs} &
    10$^9$Hz &
    2$\cdot$10$^{9}$Hz\\
    $\Gamma_f$ &
    Spectral Modes &
    30 \cite{sinclair2014spectral} &
    300 &
    3000\\
    $\eta^{\text{shift}}$ &
    Shift+Filter Efficiency &
    70\% \cite{Pan-PolarizBinToTimebin} &
    95\% &
    99\%\\
    \hline
    \multicolumn{5}{c}{Quantum Channel}\\
    \hline
    $\eta^{\text{BSM}}$ &
    BSM Efficiency &
    50\% \cite{BSMeff} &
    50\% &
    75\%\\
    $\eta^{\text{DET}}$ &
    Detection Efficiency &
    95\% \cite{SNSPDeff} &
    99\% &
    99.9\%\\
    $\alpha$ &
    Optical Fiber Attenuation (db/km) &
    0.2 &
    0.146 &
    0.146\\
    \hline
    \multicolumn{5}{c}{Router and Buffer}\\
    \hline
    $\eta^{\text{BUFF}}$ &
    Buffer Efficiency &
    30\% \cite{RbGEM} &
    90\% &
    99\%\\
    $t^{\text{BUFF}}_{\text{opt}}$ &
    Buffer Optical Coherence Time &
    30ns \cite{T2Rubidium} &
    100ns &
    500$\mu$s\\
    $t^{\text{BUFF}}_{\text{spin}}$ &
    Buffer Storage Time &
    1ms \cite{RbGEM} &
    100ms &
    500ms\\
    $\eta^{\text{MAP}}$ &
    $\text{Phot. Qubit}\!\rightarrow\!e^{-}$ Efficiency &
    6\% \cite{nemoto2014cavitymapping} &
    50\% &
    99\%\\
    \hline
    \multicolumn{5}{c}{Quantum State Transfer}\\
    \hline
    $\eta^{\text{POL}}$ &
    T.B.Q.$\rightarrow$Pol.Q. Efficiency &
    90\% \cite{Pan-PolarizBinToTimebin} &
    99\% &
    99\%\\
    $\eta^{\text{QFC}}_{637}$ &
    $\lambda_{\text{Buffer}^{-}}\!\rightarrow \lambda_{\text{NV}^{-}}$ QFC Efficiency &
    43\% \cite{Pan-PolarizBinToTimebin} &
    70\% &
    99\%\\
    \hline
    \end{tabular}
    \caption{Detailed efficiency parameters of all platforms utilized in the multi-platform network.}
    \label{tab:params1}
    \end{table} 

\section{Entanglement Generation Protocol Over the Network}\label{sect:III}
     
     As discussed thus far, the operation of the proposed network relies on successful entanglement swapping operations across the whole ARC and subsequent storage of the successful modes into the QR's long-live internal memory for further interconnection between distant QRs. An important distinction between the QR's memory and those QMs that compose the ARC must again be pointed out: due to their ensemble-based nature, the QMs contain multiple absorbers that allow, in turn, storage of photonic qubits emitted from the source at different times; the QR's internal memory, even though exhibiting much longer storage times and on-demand retrieval, does not offer the same versatility. The constraint imposed by this distinction is that, in order for the mapping into the QR's internal memory to take place, it is imperative that the entanglement swapping across the ARC has been successful. Due to the probabilistic nature of this event, the heralding signal from each individual elementary link must reach the QR node so that a decision between discarding a photonic qubit or storing it into the QR can be made.
     
     \begin{figure*}[ht]
	 \centering
	 \includegraphics[width=0.8\linewidth]{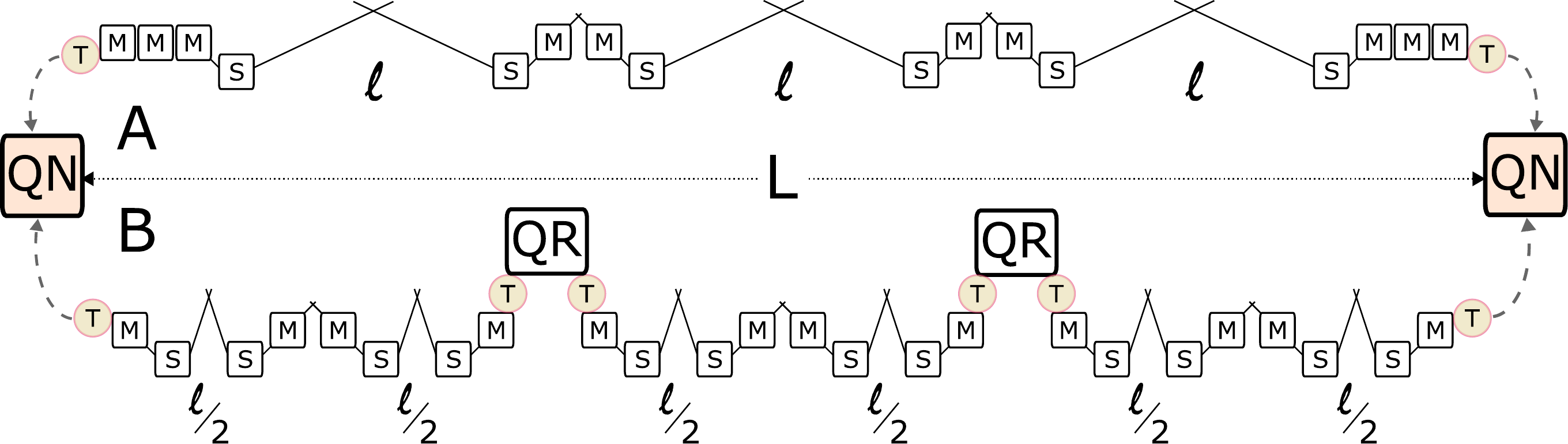}
	 \caption{Configurations A (top) and B (bottom) for the proposed multi-platform network. In A, the inclusion of extra QMs at the edges of the ARCs enforces synchronous arrival of photonic qubits and complete ARC information at the QRs connected to the ARC while maintaining the elementary link $\ell$ parameterized by the QM's storage time. In B, the length of the elementary link is reduced by a factor $\xi=2$ to enforce synchronous arrival within the ARC; although extra QMs are not necessary at the edges of the ARC, multiple ARCs are required to cover the same distance as in A.}
 	 \label{fig:sync_solAB}
     \end{figure*}
     
     It is useful to consider an example: an ARC composed of three elementary links, each 20km-long, and connected at either edge to a QR, as depicted in Fig. \ref{fig:AFC-basedARC}; furthermore, let us assume that the electron-spin-to-internal-QR-memory swap time ($t^{^{13}\text{C}}$) is $500\mu s$ and the EPPS rate ($R^{\text{EPPS}}$) is $10^8$Hz, all parameters according to the first column of Table \ref{tab:params1}. In this scenario, the photonic qubits are re-emitted from the QMs at the edges of the ARC after a time $t^{\text{RT}}$ (RT for round-trip time within an elementary link) proportional to the length of the elementary link and in accordance with the maximum storage time of the QMs ($t^{\text{AFC}}_{\text{max}} = t^{\text{RT}} = \tfrac{\ell}{c/n_r}$). Without loss of generality, we consider that the network operation starts at $t=0$ so that, at time $t = t^{\text{RT}}$, only the information associated with the BSM performed in the elementary link at the edge of the ARC is available, which is in accordance with the FFSMM operation, so that the successful spectral mode can be isolated. However, at this time, the information about whether or not the remaining two elementary links in the ARC have also been successful is not yet available; in fact, it is straightforward to show that a total time $(n-1)t^{\text{RT}}$ -- where $n$ is the number of elementary links composing the ARC -- is required for the arrival of this information at the QR.
     
     The QR is, then, faced with a choice: store this qubit in its internal memory, or discard it. For simplicity, let us assume that, at $t = t^{\text{RT}}$, i.e., the time corresponding to the first temporal mode emitted from the EPPS being re-emitted from the QM, the QR chooses to store. Recall, now, that the swapping operation from the electron spin to the QR's internal memory takes $500\mu s$, which means the QR will be unavailable for the next $m = 500\cdot10^{-6}\cdot10^8 = 5\cdot10^4$ temporal modes re-emitted from the QM; not only that, but there is no certainty that the temporal mode that has been stored corresponds to an overall successful entanglement swapping operation across the entire ARC. Hence, the optimal operation of the router-based network is dependent on the availability of the full ARC's information.
     
     This constraint is in contrast with the one imposed at each elementary link: there, the maximum achievable storage time of the QM ($t^{\text{AFC}}_{\text{max}}$) determines the length of the elementary link through the relation $t^{\text{AFC}}_{\text{max}} = t^{\text{RT}} = \tfrac{\ell}{c/n_r}$ because a FFSMM operation based on the remote BSM result within the elementary link is the only requirement to achieve interconnection between concatenated elementary links. Interconnection between QRs, on the other hand, cannot rely simply on the information associated to the elementary link closest to it (at the edge of the ARC) but, rather, on the information of the entire ARC. The hierarchical difference between elementary links and ARCs becomes clear: as the EPPSs, QMs, BSMs, and FFSMMs are connected to form an elementary link, its operation depends on information contained within the elementary link; as the elementary links are concatenated to form an ARC, its operation depends on the information contained within the ARC, i.e., the aggregate of each elementary link's information. Since the ARC communicates with the QRs through its edges, for proper operation, all the information from all elementary links must be available at the edges simultaneously with the retrieval of its corresponding photonic qubit.
     
     The total time required for the full ARC information to reach the edges corresponds to $t^{\text{ARC}} = \tfrac{L}{c/n_r} = \tfrac{n\ell_{\text{max}}}{c/n_r}$. This creates a synchronism issue across the ARC since the photonic qubits are retrieved from the QMs at the edges of the ARC earlier than when the ARC information is available; as already established, the ARC information is required simultaneously with the retrieval of its corresponding photonic qubit so that the quantum router can determine whether to store the latter into its internal memory or discard it when it does not correspond to a string of successful entanglement swapping operations. So far, the buffer has been considered so that the connectivity between two adjacent ARCs is achieved; however, under auspicious conditions, the role of mediating the flow of information between the ARC and the QR can also be assigned to the buffer, as will be described.
     
     As previously mentioned, and depicted in Table \ref{table:requirements}, the buffer operates under a temporal multiplexing protocol such that photonic qubits generated at different times (different temporal modes) can be stored in the buffer and extracted on-demand. The total number of such temporal modes that can be stored ($m_{\text{temp}}$) is determined by the product between the optical coherence time ($t^{\text{BUFF}}_{\text{opt}}$) of the buffer and the emission rate ($R^{\text{EPPS}}$) of the EPPSs according to $m_{\text{temp}} = t^{\text{BUFF}}_{\text{opt}}\cdot R^{\text{EPPS}}$. Provided $t^{\text{BUFF}}_{\text{opt}}$ is sufficiently high, the number of modes stored in the buffer can be such that the probability that at least one of these corresponds to a successful entanglement distribution across the ARC is close to unity. Furthermore, the maximum storage time in the buffer is determined by its spin coherence time ($t^{\text{BUFF}}_{\text{spin}}$), which, again, if sufficiently high, can cover the time necessary for the total ARC information to reach the QR node to ensure proper operation -- namely, $t^{\text{BUFF}}_{\text{spin}} \geq (n-1)t^{\text{RT}}$. Unfortunately, near-term parameters for buffers, as depicted in Table \ref{tab:params1}, do not match the requirements posed above, specifically regarding $t^{\text{BUFF}}_{\text{opt}}$, which is in the order of nanoseconds. Therefore, the buffer cannot mediate the synchronism between QR and ARC so that other network configurations that overcome this issue must, then, be identified; we hereby propose and discuss two such possible alternatives. The first consists of including $n-1$ QMs at the edges of the ARC so that the photonic qubits are stored and retrieved multiple times until their corresponding ARC information becomes available. The second involves shortening the elementary link length so that the storage time $t^{\text{AFC}}_{\text{max}}$ of a single QM is enough to retrieve the photonic qubit simultaneously with the arrival of the ARC information. In the following, the advantages and shortcomings of each configuration, both of which are pictorially depicted in Fig. \ref{fig:sync_solAB}-A and -B, respectively, are scrutinized allowing for numerical results to be drawn in Section \ref{sect:IV}.
     
     \subsubsection{Maximum Elementary Link Length and Extra QMs} \label{solution1}
     
     The blueprint of this alternative is depicted in Figure \ref{fig:sync_solAB}-A and involves the inclusion of QMs at the edges of the ARCs. The number of extra QMs required to achieve synchronism scales with the number of elementary links that compose the ARC as $2(n-1)$. The very characteristic of the AFC-based QM previously mentioned is taken into account here, i.e., the fact that multiple photonic qubits, emitted at different times by the EPPSs, can be simultaneously stored in the QM and retrieved at corresponding times given by $t^{\text{AFC}}_{\text{max}}$.
     
     The immediate advantage of this configuration is also its shortcoming: the fact that the structure of the network is preserved except for the inclusion of an extra $\left(\eta_{\text{AFC}}\right)^{(n-1)}$ factor in the expression of $\eta^{\text{QR}}$:
     \begin{equation}
        \label{eq:pbuf_sync}
        \eta^{\text{QR}}\!=\!\left(\eta^{\text{AFC}}\right)^{n\!-\!1}\eta^{\text{BUFF}}\eta^{\text{QFC}}_{637}\eta^{\text{POL}}\eta^{\text{MAP}}\eta^{^{13}\text{C}}.
     \end{equation}
     This means that the length of the elementary links is maintained to match $t^{\text{AFC}}_{\text{max}}$, the maximum achievable pre-set storage time of the QMs. Furthermore, it is straightforward to compute the total number of QMs required for this configuration to be $4n-2$, in case the connectivity between nodes separated by a distance L, proportional to the elementary link length $\ell_{\text{max}}$ according to $L=n\ell_{\text{max}}$, is considered.
     
     \subsubsection{Shorter Elementary Link Length}
     
     An alternative approach to the inclusion of extra QMs at the edges of the ARC is shortening the elementary link length, the blueprint for which is depicted in Figure \ref{fig:sync_solAB}-B. Here, the length $\ell_{\text{max}}$, that corresponds to a maximum pre-set storage time $t^{\text{AFC}}_{\text{max}}$, is used as reference and shortened by a factor $\xi\in\mathbb{N}|\xi\geq2$; the factor $\xi$, by definition, also corresponds to the maximum number of elementary links that can be included in a single ARC. This way, the maximum pre-set storage time of the QMs at the edges of the ARC, and connected to the QRs, automatically correspond to the time necessary for the entire ARC information to arrive at its edges.
     
     Although the extra $\left(\eta^{\text{AFC}}\right)^{n\!-\!1}$ term for $\eta^{\text{QR}}$ is not included in this configuration (in contrast with the previous one), the ARC length is limited and, thus, more ARCs are required to interconnect nodes separated by a given distance; although not the immediate focus of the analysis, this could potentially lead to reduced \textit{fidelities} of the distributed EPR-pairs. Using the same reference as before (nodes separated by $L=n\ell_{\text{max}}$), it is possible to show that this configuration makes use of $2\xi n$ QMs. Not only that, but $\xi n-1$ extra QRs are also required for the complete interconnection of two end-nodes separated by a distance $L$. Even though the overall number of resources increases, since the distance between nodes shortens, the propagation losses, that scale exponentially with the distance, will also be reduced within elementary links.
     
     \subsubsection{Summary and Protocol}
     
     Independently of the network configuration chosen, synchronism between ARCs and QRs at its edges must be enforced for the proper operation of the proposed network. Configuration A allows longer distances to be covered with fewer resources, but might become impractical depending on the parameter regime that is considered, with a special dependence on $\eta^{\text{AFC}}$. Configuration B makes use of more resources, but might be more practical for a near-term implementation since the propagation losses are diminished within an elementary link. Either way, it is clear that, depending on the specification and constraints, one configuration may achieve higher entanglement distribution rates than the other. This can lead to interesting optimization problems, which, although not posed and solved here, are addressed in Section \ref{sect:IV}, where numerical results are presented. We summarize the advantages and drawbacks of each proposed configurations in Table \ref{table:synchronism}.
     
     \begin{table}[!htb]
     \centering
     \caption{Configurations A and B Summary} \label{table:synchronism}
     \setlength\tabcolsep{5pt}
     \footnotesize\centering
     \smallskip 
     \begin{tabular}{c|c|c}
     Parameters & Configuration A & Configuration B\\
     \hline
     \hline
      \# QMs & $4n-2$ & $2\xi n$\\
      \hline
      \# QRs & \makecell{Depends on ARC length} & $\xi n -1$ required\\
      \hline
      Losses & \makecell{Potentially higher, both due\\ to a longer elementary link length\\and due to an extra $\left(\eta^{\text{AFC}}\right)^{\left(2(n-1)\right)}$\\efficiency factor.} & \makecell{Potentially smaller, since elementary\\ link length is shorter, even though\\more resources are required.}\\
     \hline
     Fidelity & \makecell{Potentially higher, since decoherence\\is limited for photonic qubits\\ propagating in the fiber.} & \makecell{Potentially smaller, as the number\\of devices between two nodes of\\the network increases.}\\
     \hline
     \end{tabular}
     \end{table}
     Consider, now, a synchronous network employing either configuration as detailed above. Under these conditions, we can introduce a parameter $\tau$ corresponding to a time interval within which entanglement must be distributed to the end nodes. This parameter, named the \textit{cut-off time}, is essential for the numerical analysis of the rates achieved in the proposed network. $\tau$ can be formally defined once all the steps necessary for entanglement distribution have been considered, which can be subdivided into two parts: in the first one, all the elementary links within the ARCs generate entanglement attempting to cover each ARC's length, which, ultimately, culminates in the storage inside a QR's memory; in the second, QRs perform entanglement swapping operations to deliver entanglement to any pair of nodes in the network. Due to the delays for synchronization, transferring the quantum states to the QR's internal memory, and for the subsequent entanglement swapping operations, the value of $\tau$ is limited from below by $t^{\arc} + t^{^{13}C} + t^{\text{CNOT}}$. On the other hand, due to decoherence during storage, quantum states cannot be stored for longer than $t^{\text{NV}}$. Hence, $\tau$ is also limited from above by $t^{\text{NV}}$, such that the following regime is imposed:
     \begin{equation}
        \label{eq:tau_lower}
        t^{\text{NV}} \geq \tau \geq t^{\arc} + t^{^{13}C} + t^{\text{CNOT}}.
     \end{equation}
    
     The relevance of the parameter $\tau$ becomes clear when one takes into account the synchronous operation of the multi-platform network: the long lifetime of the $^{13}$C memory allows the QRs to store states that have been successfully distributed while entanglement distribution operations are taking place across the multi-platform network; this leads to an increase in end-to-end entanglement distribution due to the combination of spectral (within the ARCs) and temporal (between QRs) multiplexing. In this context, if the quantum states are stored for too long, the fidelity of the states will diminish due to decoherence. Conversely, if all the routers perform a CNOT-based BSM -- leading to an entanglement swapping operation -- too soon, the internal memories might not yet be loaded. To avoid issue and simplify the operation of the network, all of the intermediate quantum routers can be set to perform entanglement swapping operations after $\tau$ seconds. Therefore, choosing the value of $\tau$ has a great impact on the operation of the network, which is analyzed in the next section for the different network configurations herewith proposed.
    
     To conclude this section, we define the entanglement generation protocol over the proposed network, i.e., the buffered-router-assisted automated repeater chain, or \textbf{ARC-R}: 
     \begin{itemize}
         \item According to the network operator's master clock, the elementary links that compose the ARCs attempt entanglement distribution through frequency-multiplexed entanglement swapping.
         \item Based on local elementary link information, the FFSMM operation is performed upon states stored in the QMs so that concatenation of elementary links is achieved within an ARC.
         \item After information of a successful entanglement distribution within an ARC reaches its edges, respective quantum states are stored in the buffers.
         \item The QRs receive the stored quantum states, when available, after the quantum state transfer chain and store them in their internal quantum memories.
         \item CNOT operations are performed within the QRs after a pre-determined time $\tau$, achieving ARC-R quantum connectivity and, thus, delivering entanglement to any set of the network's nodes.
     \end{itemize}.

     \section{Network Rate Analysis}\label{sect:IV}

     In order to provide a clear analysis of the parameter regimes and potential of the proposed network architecture, the entanglement distribution rate must be evaluated as a function of the total length between two end-nodes. The analysis focuses, first, on the features of an ARC that connects two adjacent QRs. For this, we start by considering only Configuration A for enforcing ARC synchronism; we also analyze a well-established repeater chain that makes use of a single platform as a comparison reference. Building up on these results, it is possible to define entanglement distribution rates over an ARC-R of arbitrary length. Then, we introduce comparative analyses to prove that the inclusion of the buffer is beneficial in the proposed network. Configurations A and B are also compared. For all cases, near- and long-term results are presented.
     
     \subsection{Automated Repeater Chain}
     
     In the initial ARC analysis, the adjacent QRs are connected by means of an ARC of fixed elementary link length and extra QMs for synchronism. This way, the elementary link length matches the maximum storage time permitted by the QMs, i.e., $\ell=t^{\text{AFC}}_{\text{max}}\tfrac{c}{n_r}$. We can, now, reintroduce the parameter $P^{\text{AFC}}_{\text{gen}}$, that encapsulates the probability of generation of entanglement within an elementary link, and is such that:
     \begin{equation}
       P^{\text{AFC}}_{\text{gen}}=\left(1-\left[1-e^{-\alpha\ell}\eta^{\text{BSM}}\left(\eta^{\text{DET}}\right)^2\right]^{\Gamma_f}\right)\times \left(\eta^{\text{AFC}}\eta^{\text{shift}}\right)^2
       \label{eq:P_elemLink}
     \end{equation}
     This allows us to derive the following \textbf{Lemma \ref{lemma:I}}, which determines the entanglement distribution rate between two adjacent QRs through the proposed ARC depicted in Fig. \ref{fig:sync_solAB}:
    
     \begin{lemma}
     \label{lemma:I}
        The EPR-pair generation rate per second between two quantum routers that are separated by an AFC-based ARC with n elementary links is given by
        \begin{equation}
        R^{\text{ARC}} \geq \omega_{\text{try}}^{\text{EPPS}}\left(\eta^{\text{QR}}\right)^2\left(P^{\text{AFC}}_{\text{gen}}\right)^n\left(\eta^{\text{BSM}}\right)^{n-1},
        \end{equation}
        where
        \begin{equation*}
        \omega_{\text{try}}^{\text{EPPS}} := \eta^{\text{EPPS}}R^{\text{EPPS}}.
        \end{equation*}
     \end{lemma}
     
     \begin{proof}
     For the AFC-based ARC, the frequency of the EPPS is $R^{\text{EPPS}}$ per second and its efficiency is $\eta^{\text{EPPS}}$. Hence, the total number of effective attempts per second is $\omega_{\text{try}}^{\text{EPPS}} = \eta^{\text{EPPS}}R^{\text{EPPS}}$. For generating an EPR-pair between two nodes, separated by an AFC-based ARC with $n$ elementary links, all of the $n$ elementary EPR-pairs have to be created successfully and all of the entanglement swap operations have to be performed successfully in a single attempt. Moreover, within the same attempt, the photons need to be transferred into the NV-based QRs. Probability of a successful creation of an elementary EPR-pair per attempt is given by $P^{\text{AFC}}_{\text{gen}}$ and the success probability of the entanglement swap operation is $\eta^{\text{BSM}}$; from Eq. \ref{eq:pbuf_sync}, we get that the efficiency of transferring the states from the AFC memory into the QR's internal memory is $\eta^{\text{QR}}$. Hence, the probability $P_{\text{gen}}^{\text{ARC}}\left(n\right)$ of generating an EPR-pair (per attempt) between two nodes which are separated by an ARC with $n$ elementary links is given by:
     \begin{equation}
         P_{\text{gen}}^{\text{ARC}}\left(n\right) = \left(\eta^{\text{QR}}\right)^2\left(P^{\text{AFC}}_{\text{gen}}\right)^n\left(\eta^{\text{BSM}}\right)^{n-1},
         \label{eq:lemma1_prob}
     \end{equation}
     where $\omega^{\text{EPPS}}_{\text{try}}$ attempts per second are possible. This implies that the EPR-pair generation rate is given by
     \begin{equation}
         R^{\text{ARC}} \geq \omega^{\text{EPPS}}_{\text{try}}\left(\eta^{\text{QR}}\right)^2\left(P^{\text{AFC}}_{\text{gen}}\right)^n\left(\eta^{\text{BSM}}\right)^{n-1},
     \end{equation}
     which concludes the proof.
    \end{proof}
    
     This result is central for the development of a generic node-to-node entanglement distribution analysis, as it provides a simple lower bound of the expected entanglement distribution rates between two adjacent QRs in the proposed ARC-R network. It is important to recall that both configurations A and B can be analyzed using this result by replacing the expression of $\eta^{\text{QR}}$ from Eq. \ref{eq:pbuf} to Eq. \ref{eq:pbuf_sync}.
     
     At this point, a comparison can be drawn, in terms of entanglement distribution rates, between two adjacent QRs in two different network architectures: the ARC-R proposed here and determined in \textbf{Lemma \ref{lemma:II}}; and an equivalent automated repeated chain based on NV nodes. The goal of this comparison is to determine whether a hybrid network such as the proposed ARC-R, that makes use of frequency-multiplexing to achieve higher rates and buffered-routers, can overcome a network composed of a single constituent; the NV has been chosen for this comparison since recent experimental demonstrations \cite{tchebotareva2019entanglement} set it as one of the most prominent candidates for the composition of a near-term quantum network. The architecture of such an NV chain is presented in Fig. \ref{fig:NVbased_ARC}, where, for comparison sake, the elementary link length is considered to be the same as in the ARC-R case, i.e., $\ell=t^{\text{AFC}}_{\text{max}}\tfrac{c}{n_r}$.

    \begin{figure}[h]
	 \centering
	 \includegraphics[width=0.8\linewidth]{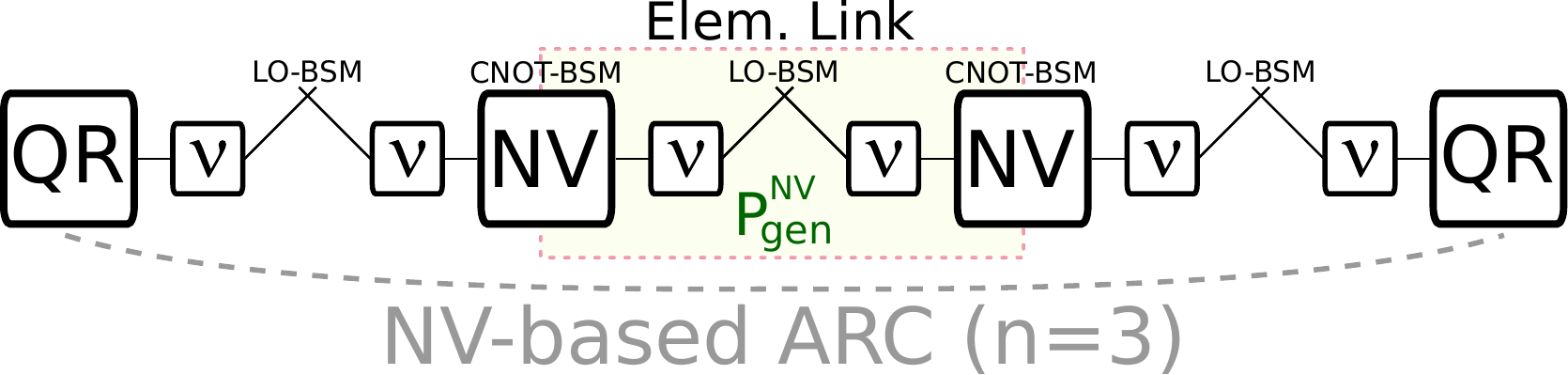}
	 \caption{Pictorial representation of similar network architecture as in Fig. \ref{fig:AFC-basedARC}, but composed only of NV nodes. $n=3$ concatenated elementary links are represented.}
	 \label{fig:NVbased_ARC}
     \end{figure}
    
     An important note is that, following Fig. \ref{fig:NVbased_ARC}, the routers are, themselves, implemented by NV devices, eliminating the requirement of the quantum state transfer step discussed in Section \ref{sect:II}; however, the photons emitted by the NVs must be frequency-converted to a wavelength suitable for long-distance optical fiber transmission \cite{tchebotareva2019entanglement}. We denote the efficiency of such conversion as $\eta^{\text{QFC}}_{\text{1588}}$, according to Table \ref{tab:params1}, and point out that the elementary links are concatenated not by a linear-optics-based BSM, but by a deterministic BSM through a CNOT gate within the NVs ($\eta^{\text{CNOT}}=1$). Analogously, a parameter that encapsulates the entanglement generation probability within an NV-based elementary link can be defined:
     \begin{equation}
        P_\text{gen}^{\text{NV}}=\left(1-\left[1-\left(\eta^{\text{QFC}}_{\text{1588}}\right)^2e^{-\alpha \ell}\eta^{\text{BSM}}\right]^{\Gamma_t}\right).
     \end{equation}
     Furthermore, due to the time-multiplexed feature of this network architecture, the parameter $\tau^{\text{NV}}$, analogous to $\tau$ defined in Eq. \ref{eq:tau_lower}, is introduced. Here, we introduce a confidence parameter $0 \leq \epsilon \leq 1$ and choose $\tau^{\text{NV}}$ such that the probability $P^{\text{NV-C}}_{\text{gen}}\left(n\right)$ of generating an EPR-pair between two end nodes of an NV-chain of length $n$, such as the one depicted in Fig. \ref{fig:NVbased_ARC}, is at least $1-\epsilon$. The exact expression of $\tau^{\text{NV}}$ is a by-product of \textbf{Lemma \ref{lemma:II}}, which determines the entanglement distribution rate of an NV-chain:
    
     \begin{lemma}
     \label{lemma:II}
        The EPR-pair generation rate per second between two nodes that are separated by an NV-chain with n elementary links is given by
        \begin{equation}
        R^{\text{NV-C}} \geq
        \frac{1}{\tau^{\text{NV}}}\left(1-\left(1-P_{\text{gen}}^{\text{NV}}\right)^{\omega_{\text{try}}^{\text{NV}}\left(\tfrac{\tau^{\text{NV}}}{2} - \tilde t^{\text{trans}}\right)}\right)^n,
        \end{equation}
        where
        \begin{equation*}
        \begin{split}
        \tilde t^{\text{trans}} &:= t^{^{\text{13}}C}+t^{\text{CNOT}},\\
        \omega_{\text{try}}^{\text{NV}}&:=\frac{c}{\ell n_r}.
        \end{split}
        \end{equation*}
     \end{lemma}	
	
     \begin{proof}
     In the NV-chain, each of the intermediate repeater nodes sends a photon to the remote BSM station and waits for the heralded signal to arrive; another attempt starts after this waiting period. If the length of the elementary link is $\ell$, then the waiting time between two successive attempts is at least $t^{\text{AFC}}_{\text{max}}$ seconds, since $\ell = t^{\text{AFC}}_{\text{max}}\tfrac{c}{n_r}$; hence, the total number of attempts per second is $\omega_{\text{try}}^{\text{NV}}=\frac{c}{\ell n_r}$. Note that, due to the availability of a single communication qubit, an NV-based intermediate quantum repeater in an ARC can only attempt to generate EPR-pairs with one of its neighbours at a time.
     
     According to the EPR-pair distribution protocol for the NV-chain, all of the intermediate repeaters perform the CNOT-based entanglement swapping operations after $\tau^{\text{NV}}$ seconds. During that time period, all of the intermediate repeaters must successfully create EPR-pairs, store them inside the $^{13}\mathrm{C}$ internal memory, and perform the CNOT operations to close the entanglement swapping chain. The time necessary to transfer the generated EPR-pairs to the $^{13}\mathrm{C}$ memory is $t^{13C}$ seconds and the time to perform the entanglement swapping operation is $t^{\text{CNOT}}$ seconds. Hence, among the $\tau^{\text{NV}}$ seconds, the ARC should reserve $\tilde t^{\text{trans}} := t^{^{13}C} + t^{\text{CNOT}}$ seconds for transferring the qubits and for entanglement swapping operations. Also due to the availability of a single communication qubit, the intermediate nodes alternate between the links connected to its either sides, leaving a total duration of $\left(\tfrac{\tau^{\text{NV}}}{2} - \tilde t^{\text{trans}}\right)$ to attempt EPR-pair generation per node per elementary link. During this period, the nodes can attempt $\omega_{\text{try}}^{\text{NV}}\left(\tfrac{\tau^{\text{NV}}}{2} - \tilde t^{\text{trans}}\right)$ times the creation of one EPR-pair, where the success probability per attempt is given by $P_{\text{gen}}^{\text{NV}}$. This implies that the probability of generating at least one EPR-pair during the time window $\tau^{\text{NV}}$, per elementary link, is 
     \begin{equation}
     P^{\text{NV-C}}_{\text{gen}}\left(n\right)=\left(1-\left(1-P_{\text{gen}}^{\text{NV}}\right)^{\omega_{\text{try}}^{\text{NV}}\left(\tfrac{\tau^{\text{NV}}}{2} - \tilde t^{\text{trans}}\right)}\right)^n.    
     \end{equation}
     From this expression, it is possible to extract a close form for $\tau^{\text{NV}}$, as follows:
     \begin{equation}
        \label{eq:tau_nv}
        \tau^{\text{NV}} = 
        \left(\frac{2}{\omega^{\text{NV}}_{\text{try}}} \frac{\log \left(1-(1-\epsilon)^{\frac{1}{n}}\right)}{\log \left(1-P^{\text{NV}}_{\text{gen}}\right)} + \tilde{t}^{\text{trans}}\right).
    \end{equation}
    Since $\tau^{\text{NV}}$ is upper-bounded by $t^{\text{NV}}$ according to Eq. \ref{eq:tau_lower}, in case the expression in Eq. \ref{eq:tau_nv} exceeds $t^{\text{NV}}$, we set $\tau^{\text{NV}} = t^{\text{NV}}$.
    
    There are, in total, $\frac{1}{\tau^{\text{NV}}}$ time windows of duration $\tau^{\text{NV}}$ per second. Hence, the expected EPR-pair generation rate is given by:
     \begin{equation}
     R^{\text{NV-C}} \geq \frac{1}{\tau^{\text{NV}}}\left(1-\left(1-P_{\text{gen}}^{\text{NV}}\right)^{\omega_{\text{try}}^{\text{NV}}\left(\tfrac{\tau^{\text{NV}}}{2} - \tilde t^{\text{trans}}\right)}\right)^n,
     \end{equation}
    which concludes the proof.
    \end{proof}
     
    In possession of \textbf{Lemmas \ref{lemma:I}} and \textbf{\ref{lemma:II}}, it is possible to, now, compare the entanglement distribution rates between two adjacent QRs connected either by the AFC-based ARC (\textbf{Lemma \ref{lemma:I}}) or the NV-chain (\textbf{Lemma \ref{lemma:II}}). For this analysis, the elementary link lengths are chosen to be the same and determined by the maximum QM storage time ($\ell=t^{\text{AFC}}_{\text{max}}\tfrac{c}{n_r}$), which correspond to 20 and 60km for near- and long-term parameters, respectively; in order to simplify the result visualization, ideal parameters are not considered. The results of the comparison, in terms of achievable entanglement distribution rates, are depicted in Fig. \ref{fig:NV vs AFC} as a function of the number of concatenated elementary links forming the repeater chain. For these and further results, the confidence parameter is set as $\epsilon = 0.05$, i.e., the time before a CNOT operation is performed within the nodes is such that an EPR-pair has been successfully distributed by the elementary link with 95\% certainty.
     
    \begin{figure}[ht]
    	\centering
        \includegraphics[width=0.8\linewidth]{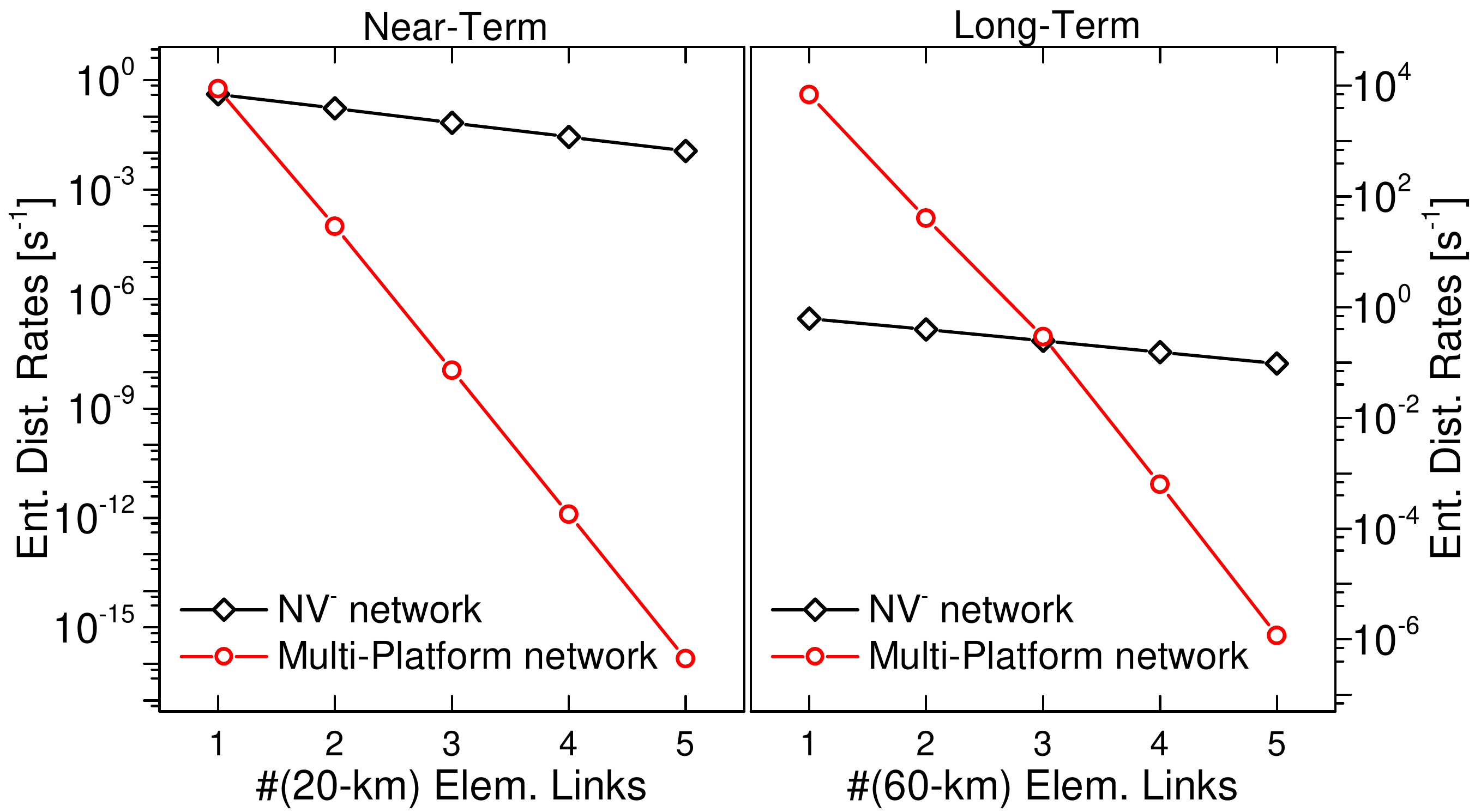}
	    \caption{Entanglement distribution rate across a homogeneous NV$^{-}$-base network and a multi-platform network versus number of concatenated elementary links using near- and long-term parameters.}
        \label{fig:NV vs AFC}
    \end{figure}
    
    The results of Fig. \ref{fig:NV vs AFC} indicate that the scaling of entanglement distribution rates, for both near- and long-term parameters, as the number of concatenated elementary links increases, is worse for the proposed multi-platform network. This analysis is interesting since it demonstrates that the boost in EPR-pair generation induced by the frequency-multiplexing capacity of the AFC-based ARC cannot completely compensate for the losses as the ARC length increases. On the other hand, this boost seems to provide a positive offset to the rate and, thus, make the multi-platform network several orders of magnitude more viable in case the number of concatenated elementary links between two QRs is kept within a certain number. This behavior stems from the fact that the large number of distinct spectral modes utilized in the AFC-based ARC mitigates the inefficiency imposed by the non-deterministic entanglement swapping operations in the linear-optics-based BSM within, but not between, elementary links; such probabilistic operations become a more dominant factor as the number of concatenated elementary links increases, resulting in worse scaling when compared to the NV-chain, which performs near-deterministic entanglement swapping through a CNOT gate. It is important to note that these results are highly dependent on the parameters of Table \ref{tab:params1}, which have been chosen based on current literature parameters and projections for the future of the physical systems that constitute the multi-platform network: changes to these parameters will result in different behaviors. We point this out since these results will be used henceforth to impose constraints on the subsequent analysis.
     
    For both the near- and long-term parameters considered in the analysis of Fig. \ref{fig:NV vs AFC}, the multi-platform network outperforms the NV-chain for a single elementary link ARC; for long-term parameters, this regime extends up to three elementary links. This is an important result that will impact on the analysis of the next section, where multiple ARCs are connected through QRs to compose the ARC-R. We, thus, fix the number $n$ of elementary links that will be analyzed in the ARC-R case based on the results of Fig. \ref{fig:NV vs AFC} as one ($n=1$) and two ($n=2$) for near- and long-term parameters, respectively. This translates into QRs spaced by 20km for the near-term multi-platform network, while for a  long-term ARC-R, the QRs would be spaced by 120km, when considering configuration A. Such distances are within the values expected for future quantum networks spanning continents \cite{wehner2018Qinternet}. Furthermore, the choice for long-term parameters is made based on a compromise between longer ARCs and higher rates: by choosing 2 elementary links over 3, the ARCs reach 120 instead of 180 km, but achieve over two orders of magnitude higher entanglement distribution rates. In terms of resources required by a potential ARC-R, it is easy to see that the overall number of QRs will depend directly on the length of the ARCs and, while going from 60 to 120km ($n=1$ to $n=2$) halves the number of required QRs, going from 120 to 180km ($n=2$ to $n=3$) reduces the number of required QRs by a mere third.
    
    \subsection{Buffered-Router-Assisted Automated Repeater Chain}

    The central result of this Section is the node-to-node entanglement distribution rates mediated by the ARC-R. One interesting observation is the fact that, differently from \textbf{Lemma \ref{lemma:I}} -- where the cut-off time was not required to compute the rates --, but similarly to \textbf{Lemma \ref{lemma:II}}, $\tau$ plays an important role in the estimation of the entanglement distribution rates for the ARC-R. As previously defined in Section \ref{sect:III}, the proposed entanglement generation protocol requires that all intermediate quantum routers (between the end-nodes) perform the entanglement swapping operations after a pre-defined $\tau$ seconds. As in \textbf{Lemma \ref{lemma:II}}, we introduce a confidence parameter $0 \leq \epsilon \leq 1$ and choose $\tau$ such that the probability $P_{\text{gen}}^{\text{ARC-R}}$ of generating an EPR-pair between two end nodes in an ARC-R of length $N$ is at least $1-\epsilon$. The expression of $P_{\text{gen}}^{\text{ARC-R}}$ and that of $\tau$ are both a by-product of \textbf{Lemma \ref{lemma:III}}, where we compute the achievable entanglement distribution rates between nodes connected through an ARC-R of varying length N, i.e., with N routers and N+1 ARCs.
    
    \begin{lemma}
    \label{lemma:III}
    In an ARC-R of length N, the EPR-pair generation rate per second between two end-nodes is given by
    \begin{equation}
    R^{\text{ARC-R}}\!\geq\!\frac{1}{\tau}\!\left(\!1\!-\!\left[1-P_{\text{gen}}^{\text{ARC}}\left(n\right)\right]^{\omega_{\text{try}}^{\text{EPPS}}\left(\tau- t^{\text{trans}}\right)}\right)^N,
    \end{equation}
    where 
    \begin{equation*}
        t^{\text{trans}} := (t^{^{13}C} + t^{\arc} + t^{\text{CNOT}}).
    \end{equation*}
    \end{lemma}

    \begin{proof}
    In order for an EPR-pair to be created along an ARCR, all of the intermediate quantum routers need to create and perform swapping operations within $\tau$ seconds, since an intermediate quantum router stores the EPR-pair generated in an ARC for $\tau$ seconds. According to Eq. \ref{eq:tau_lower}, $\tau$ is lower bounded by the network delay that includes synchronization, transfer to the QR's internal memory, and entanglement swapping operations, a parameter condensed into $t^{\text{trans}} := (t^{^{13}C} + t^{\text{ARC}} + t^{\text{CNOT}})$.
    
    Due to the inclusion of a buffer with longer storage time than $t^{^{13}C}$, the two ARCs connected to a QR can simultaneously attempt to generate EPR-pairs. Hence, all of the ARC's can effectively attempt to distribute entanglement for a total duration $\tau- t^{\text{trans}}$ seconds, with a rate of effective attempts equal to $\omega_{\text{try}}^{\text{EPPS}}=\eta^{\text{EPPS}}R^{\text{EPPS}}$, totalling $\omega_{\text{try}}^{\text{EPPS}}\left(\tau- t^{\text{trans}}\right)$ attempts. We now invoke the result of \textbf{Lemma \ref{lemma:I}}, in the form of Eq. \ref{eq:lemma1_prob}, to write the probability of generating at least one EPR-pair between two QR's connected through an ARC of length $n$ given $\omega_{\text{try}}^{\text{EPPS}}\left(\tau- t^{\text{trans}}\right)$ attempts: $\left(1-\left[1-P_{\text{gen}}^{\text{ARC}}\left(n\right)\right]^{\omega_{\text{try}}^{\text{EPPS}}\left(\tau- t^{\text{trans}}\right)}\right)$.
    
    For two end nodes of the ARC-R to be able to share an EPR-pair, it is necessary for each intermediate router to share an EPR-pair with its neighbour following the above procedure; in an ARC-R of length N, this corresponds to a total of N such EPR-pairs. Hence, the probability of generating one end-to-end EPR-pair within a time-window $\tau$ is:
    \begin{equation}
        P_{\text{gen}}^{\text{ARC-R}}\!=\!\left(\!1\!-\!\left[1\!-\!P_{\text{gen}}^{\text{ARC}}\left(n\right)\right]^{\omega_{\text{try}}^{\text{EPPS}}\left(\tau- t^{\text{trans}}\right)}\right)^N.
    \end{equation}
    From this expression, it is possible to extract a close form for $\tau$, as follows:
     \begin{equation}
        \label{eq:tau_arcr}
        \tau = 
        \left(\frac{1}{\omega^{\text{EPPS}}_{\text{try}}} \frac{\log \left(1-(1-\epsilon)^{\frac{1}{N}}\right)}{\log\! \left[1-P^{\text{ARC}}_{\text{gen}}\left(n\right)\right]}\!+\!t^{\text{trans}}\right).
    \end{equation}
    Since $\tau$ is upper-bounded by $t^{\text{NV}}$ according to Eq. \ref{eq:tau_lower}, in case the expression in Eq. \ref{eq:tau_arcr} exceeds $t^{\text{NV}}$, we set $\tau = t^{\text{NV}}$.
    
    Since, in a second, there are $\frac{1}{\tau}$ time-windows $\tau$, the expected EPR-pair generation rate $R^{\text{ARC-R}}$ of a buffer-router-assisted automated repeater chain is:
    \begin{equation}
     R^{\text{ARC-R}}\!\geq\!\frac{1}{\tau}\!\left(\!1\!-\!\left[1-P_{\text{gen}}^{\text{ARC}}\left(n\right)\right]^{\omega_{\text{try}}^{\text{EPPS}}\left(\tau- t^{\text{trans}}\right)}\right)^N,
    \end{equation} 
    which concludes the proof. 
    \end{proof}
    
    It is interesting to note that \textbf{Lemma \ref{lemma:III}} is an amalgam of \textbf{Lemmas \ref{lemma:I}} and \textbf{\ref{lemma:II}}, with the difference that the EPR-pair generation rate between two QRs is given by $P_{\text{gen}}^{\text{ARC}}\left(n\right)$. As already observed in Fig. \ref{fig:NV vs AFC}, there is a regime for which better results can be achieved following the architecture of the proposed multi-platform network. Another feature of the multiplicity of platforms is the presence of the buffer, which has been considered in the proof of \textbf{Lemma \ref{lemma:III}}. Determining the rates for a similar network that does not make use of such buffer unit is important so that a fair comparison between the two can be performed, which is presented in \textbf{Lemma \ref{lemma:IV}}. It is noteworthy that, due to the lack of the buffer, the expression for $\tau$ will change in \textbf{Lemma \ref{lemma:IV}} with respect to that presented in \textbf{Lemma \ref{lemma:III}}. Since the two expressions are quite similar but not equal, $\tau'$ was chosen to represent the pre-determined time before a CNOT operation within the QRs for an ARC-R destitute of buffers.
    
    \begin{lemma}
    \label{lemma:IV}
    In an ARC-R of length N destitute of buffers, the EPR-pair generation rate per second between two end-nodes is given by
    \begin{equation}
    \hspace{-5pt}R_{\text{noBuff}}^{\text{ARC-R}}\!\geq\!\frac{1}{\tau'}\!\left(\!1\!-\!\left[\!1\!-\!\frac{P_{\text{gen}}^{\text{ARC}}\!\left(n\right)}{\left(\eta^{\text{BUFF}}\right)^2}\!\right]^{\omega_{\text{try}}^{\text{EPPS}}\!\left(\!\tfrac{\tau'}{2}\!-\!t^{\text{trans}}\!\right)}\!\right)^N.
    \end{equation}
    \end{lemma}

    \begin{proof}
    In an ARC-R destitute of buffers, the QRs can only communicate to one ARC at a time, i.e., the ARC-L and ARC-R, connected to the left and right of a QR, respectively, cannot generate EPR-pairs simultaneously. Hence, each of these ARCs can effectively use half of the $\tau'$ time-window that is available. Within that time, $t^{\text{trans}}$ seconds must be reserved, which yields a net attempt time-window of $\frac{\tau'}{2} - t^{\text{trans}}$ seconds, and, furthermore, a total number of attempts, within this time-window, of $\omega_{\text{try}}^{\text{EPPS}}\left(\frac{\tau'}{2} - t^{\text{trans}}\right)$. We now invoke the results of \textbf{Lemma \ref{lemma:I}} noticing that, in an ARC-R destitute of buffers, the buffer efficiency term $\eta^{\text{BUFF}}$ does not play a role. This has an impact on the expression of $P_{\text{gen}}^{\text{ARC}}\left(n\right)$, which can be compensated for by extracting, from the former, a factor $\left(\eta^{\text{BUFF}}\right)^2$. Therefore, the probability of generating at least one EPR-pair between two QRs connected through an ARC of length $n$ destitute of buffers given $\omega_{\text{try}}^{\text{EPPS}}\left(\tfrac{\tau'}{2}- t^{\text{trans}}\right)$ attempts is: $\left(1-\left[1-\frac{P_{\text{gen}}^{\text{ARC}}\left(n\right)}{\left(\eta^{\text{BUFF}}\right)^2}\right]^{\omega_{\text{try}}^{\text{EPPS}}\left(\tfrac{\tau'}{2}- t^{\text{trans}}\right)}\right)$. The same reasoning of \textbf{Lemma \ref{lemma:III}} can now be followed, yielding:
    \begin{equation}
        \label{eq:tau_arcr_noBuff}
        \tau' = 
        \left(\frac{2}{\omega^{\text{EPPS}}_{\text{try}}} \frac{\log \left(1-(1-\epsilon)^{\frac{1}{N}}\right)}{\log \left(1-\tfrac{P^{\text{ARC}}_{\text{gen}}\left(n\right)}{\left(\eta^{\text{BUFF}}\right)^2}\right)} + t^{\text{trans}}\right);
    \end{equation}
    and
    \begin{equation}
        \label{eq:rate_no_buff}
        \hspace{-5pt}R_{\text{noBuff}}^{\text{ARC-R}}\!\geq\!\frac{1}{\tau'}\!\left(\!1\!-\!\left[\!1\!-\!\frac{P_{\text{gen}}^{\text{ARC}}\!\left(n\right)}{\left(\eta^{\text{BUFF}}\right)^2}\!\right]^{\omega_{\text{try}}^{\text{EPPS}}\!\left(\!\tfrac{\tau'}{2}\!-\!t^{\text{trans}}\!\right)}\!\right)^N.
    \end{equation}
    This concludes the proof.
    \end{proof}
    
    The purpose of the multi-platform network architecture herewith proposed is to harness the resources available in each of the different physical systems. The set goal is achieved in the regime where the complexity introduced by the connectivity between multiple platforms allows for higher entanglement distribution rates. From \textbf{Lemmas \ref{lemma:I}} and \textbf{\ref{lemma:II}} and the results of Fig. \ref{fig:NV vs AFC}, such regime is identified in a simple scenario, with two QRs connected through either an AFC-based ARC or an NV-chain, for near- and long-term parameters as up to 1 and 2 frequency-multiplexed elementary links connecting the QRs, respectively. We extend this analysis in the context of \textbf{Lemmas \ref{lemma:III}} and \textbf{\ref{lemma:IV}} by fixing these numbers, i.e., ARCs of 1 and 2 elementary links for near- and long-term parameters, respectively, and evaluating the entanglement distribution rates now between two end-nodes of a one-dimensional network of increasing length. The results, presented in Fig. \ref{fig:bufferComparison}, show achievable rates for our proposed ARC-R (following \textbf{Lemma \ref{lemma:III}}), but also those for an ARC-R destitute of buffers (following \textbf{Lemma \ref{lemma:IV}}) and for an NV-chain (following \textbf{Lemma \ref{lemma:II}}) of same length. 
    
    \begin{figure}[ht]
    	\centering
        \includegraphics[width=0.8\linewidth]{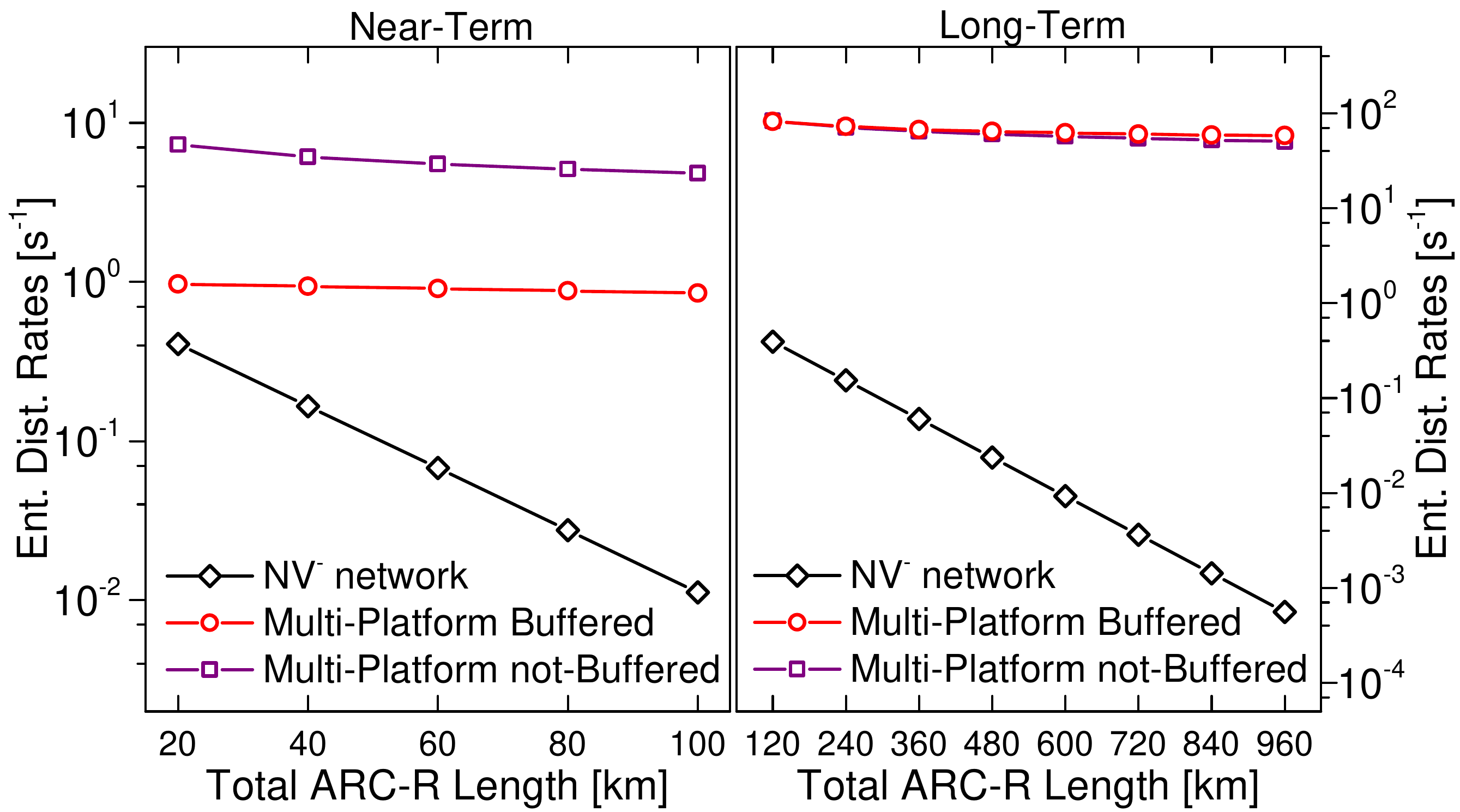}
	    \caption{Comparison between our multi-platform network with and without buffer and its homogeneous NV$^{-}$-based counterpart. Using near- and long-term parameters, entanglement distribution rate is computed versus ascending ARC-R length.}
        \label{fig:bufferComparison}
    \end{figure}
    
    The first comment that should be made from Fig. \ref{fig:bufferComparison} is that the rates computed with \textbf{Lemma \ref{lemma:III}} between two adjacent QRs connected by an AFC-based ARC match those computed with \textbf{Lemma \ref{lemma:I}} for the equivalent cases, as it is expected. Secondly, the fact that the achievable entanglement distribution rates of the homogeneous NV-chain scales badly in comparison to the multi-platform network attests the former's achievement in improving the rates albeit increasing the complexity of the network. This is a reflection of the spectral multiplexing feature of the AFC-based ARC in conjunction with the deterministic nature of entanglement swapping once the EPR-pairs distributed in a probabilistic fashion within the ARCs are mapped to the QR's internal memory. Finally, we highlight the dependency of the ARC-R's overall performance on the buffer's efficiency: in the near-term, the absence of the buffer causes the rates to increase whereas in the long-term the buffered ARC-R reaches higher rates. \textcolor{black}{The reason behind that is the higher efficiency levels achieved by the buffer infrastructure in the long-term when compared to the near-term; therefore, the buffer utility is expected to be more relevant if better performance can be achieved. The fact that a marginal gain is observed in Fig. 7-b, however, could raise the question of whether this is indeed the case. As will be clarified in Fig. 9, the entanglement distribution protocol currently considered for the proposed network, although allowing for a simplified analysis of the rates, is sub-optimal and becomes more taxing as the overall network efficiency increases. This is because, for a less efficient network, waiting for tau seconds to perform the entanglement swapping operations is less prejudicial than for a more efficient network since the probability that the routers are already loaded before tau seconds is very low for the former. Although the marginal gain of the buffered network with respect to the non-buffered network observed in Fig. 7-b is already expected for a linear network, where each router is connected to at most two ARCs, (even for an optimal entanglement distribution protocol) it is further decreased due to the shortcomings of the considered entanglement distribution protocol.} We conjecture that the impact of the buffer will be even greater in a more complex network topology, where the QRs can be connected to multiple ARCs through more than two buffers, i.e., \textcolor{black}{the more efficient the components and the more complex the network topology, the more imperative the buffer.}
    
    \begin{figure}[ht]
    	\centering
        \includegraphics[width=0.8\linewidth]{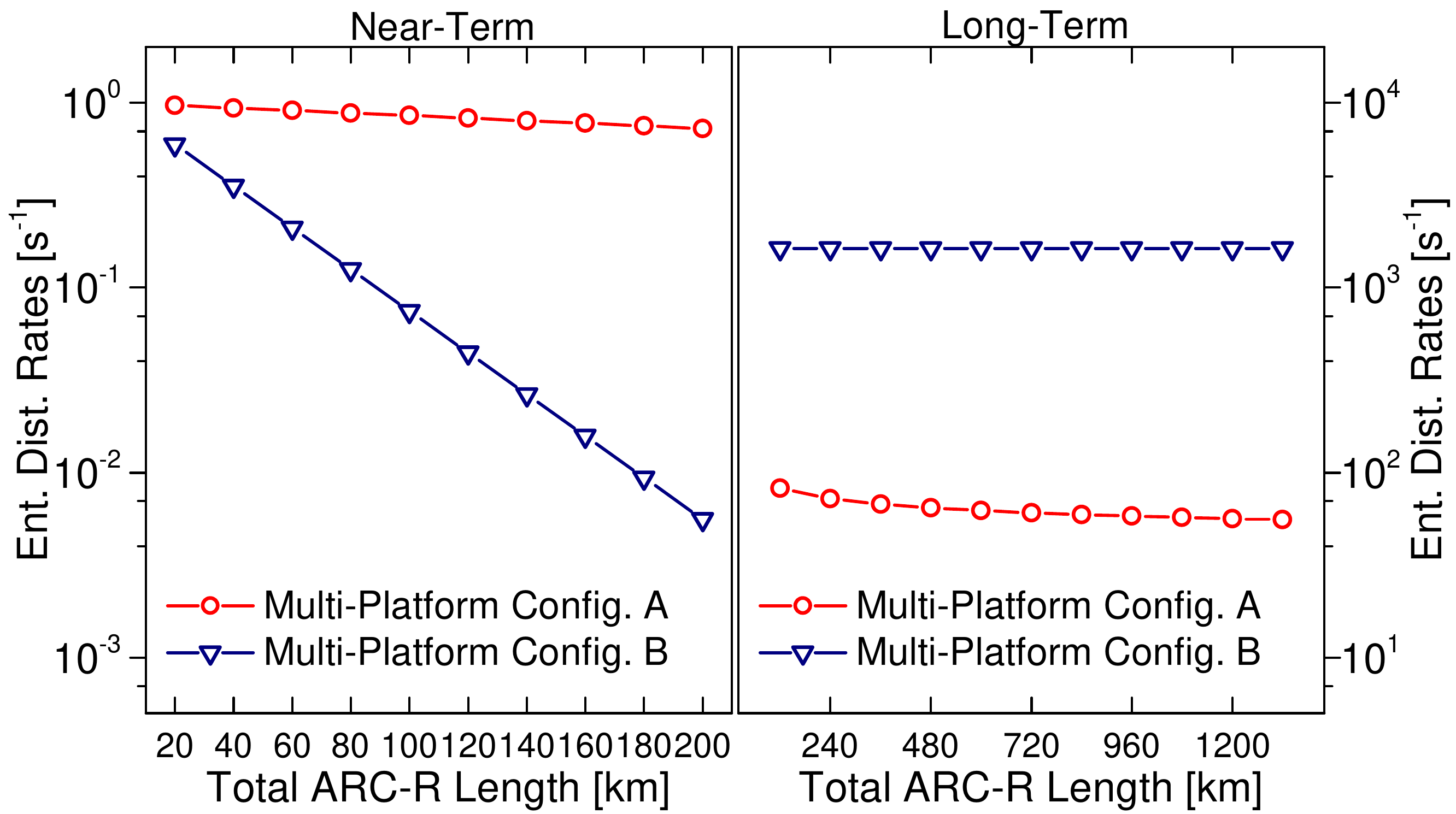}
	    \caption{Entanglement distribution rates achieved by configurations A and B for near- and long-term parameters as a function of total ARC-R length.}
        \label{fig:Config A and B}
    \end{figure}
    
    The final important entanglement distribution rate comparison of the proposed multi-platform network is regarding configurations A and B. The results, also computed as a function of increasing ARC-R length, are depicted in Fig. \ref{fig:Config A and B}. The dependence on the near- and long-term parameters can once again be observed since, in the former, configuration A outperforms B, while for the latter it is the opposite. This result is tightly associated to the projected increase in efficiencies for those operations required for the QST between the ARC and the QRs. Configuration B makes use of more QRs and, thus, more such QST operations are necessary for the same ARC-R length; thus, as these efficiencies increase in the long-term projections, this configuration achieves higher entanglement distribution rates. Simultaneously, the observation stated in Table \ref{table:synchronism} also plays an important role, since the length of the elementary links are shortened in configuration B and, thus, will allow for higher achievable rates. An in-depth comparison between configurations A and B with an optimization perspective, even though not in the scope of this work, could further elucidate the regimes where one outperforms the other. This also applies to number of elementary links composing an ARC, especially when the fidelity of the distributed EPR-pair is taken into account. The latter is discussed in the next Section, where a simplified decoherence model is employed to extract expected fidelity values achieved by the network architectures herewith discussed.
    
    To finalize this Section, we show, in the left pane of Fig. \ref{fig:tauBehavior}, the behavior of the parameter $\tau$ as a function of the number of AFC-based ARCs that compose an ARC-R. Since $\tau$ has a high dependency on the probability of distribution of entanglement within an ARC, it is also expected that, as the length of the elementary link changes, $\tau$ should decrease. Although, in the previous analyses, the length of the elementary link was always considered as a fixed parameter, we showcase this dependency in the right panel of Fig. \ref{fig:tauBehavior}, where the length of the elementary link varies between 10 and 100km irrespective of the QM's storage time. In order to simplify these results, the number of ARCs in the ARC-R is fixed at 1 and configuration A is considered. For both figures, the number of elementary links in an ARC has been fixed to 1 and 2 for near- and long-term parameters, respectively, and, in consistency with the previous results, the value of $\epsilon$ is kept at $5\%$.
    
    \begin{figure}[ht]
    	\centering
    	\includegraphics[width=0.8\linewidth]{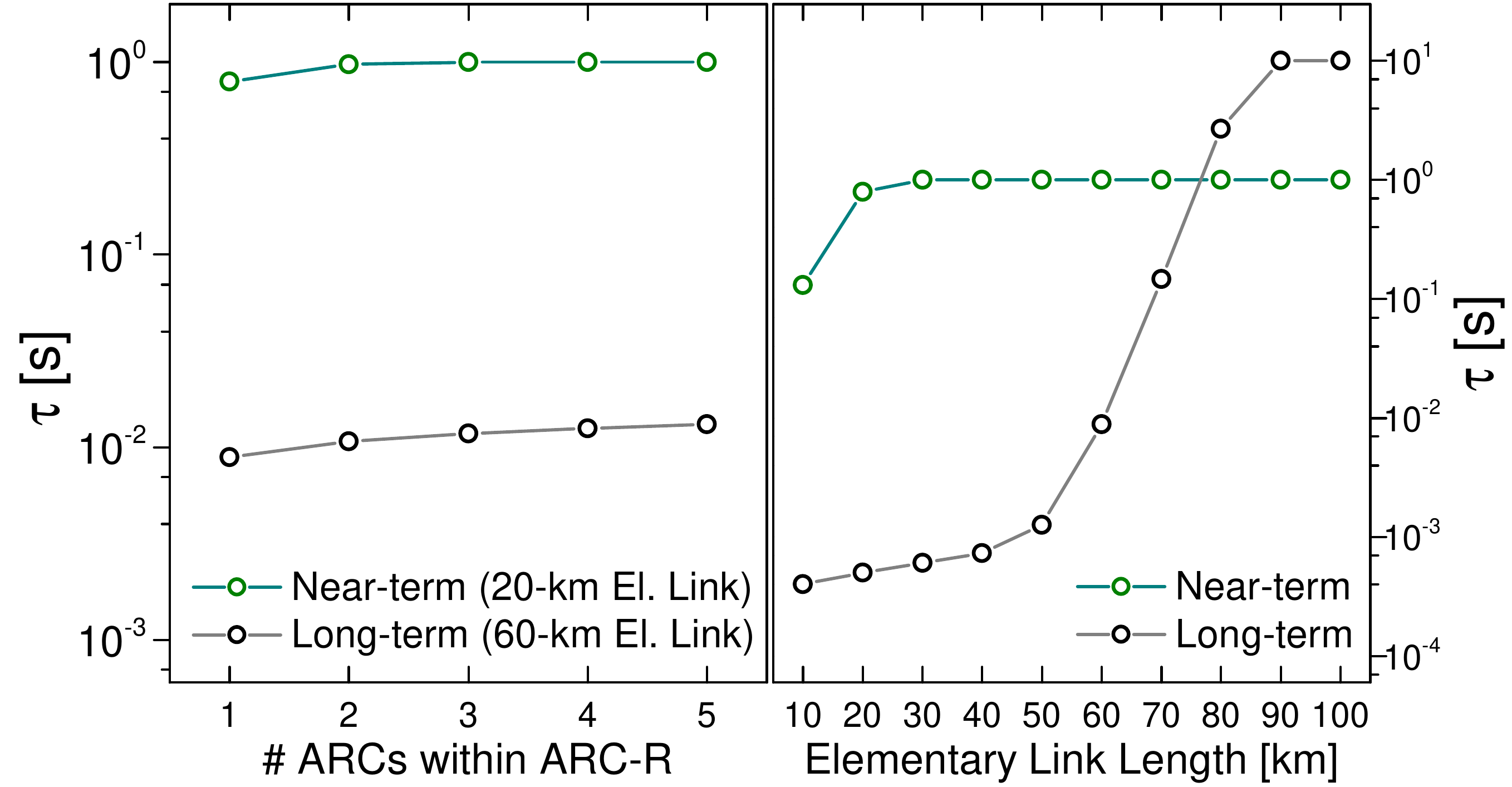}
    	\caption{Behavior of the parameter $\tau$ under different network conditions. On the left pane, $\tau$ is depicted as a function of the number of fixed-length ARCs within an ARC-R; on the right, the elementary link length is varied for an ARC-R composed of a single ARC.}
    	\label{fig:tauBehavior}
    \end{figure}
    
    The results of Fig. \ref{fig:tauBehavior} are important to highlight the fact that setting up a pre-determined time before CNOT operations in the QRs simplifies the analysis of the network but is far from being the optimal entanglement distribution protocol over a chain network. In fact, one can observe that, especially for near-term parameters, the value of $\tau$ reaches the maximum storage time of the QR's internal memory (1 second, according to the first column of Table \ref{tab:params1}). The greater drawback of this protocol is the likely creation of idle QRs, which are ready to perform an entanglement swapping operation, but are tied to the pre-determined value of $\tau$.
    
    \section{Network Fidelity Analysis}\label{sect:V}
    
    The previous Section focused on determining the expected entanglement distribution rates in the proposed multi-platform quantum network. Although achieving high-rates is an important goal of the network, ensuring that the distributed bipartite states exhibit high fidelity with respect to an ideal maximally-entangled Bell-state is a prerequisite for the usefulness of the network. In other words, achieving the highest possible rates corresponds to one aspect of the role of the quantum network, which only becomes viable in case those states that are distributed actually meet the fidelity criterion that allows the end-nodes to make use of them to implement the envisioned quantum communication and computation protocols. Here, we split the distribution rate and fidelity calculation, with the former already covered in the previous section. The latter is determined through a simplified model, where the imperfect states generated by the EPPSs are considered to be Werner states \cite{werner1989quantum} and all the processes' effect on the fidelity of the distributed states are that of a depolarizing channel \cite{king2003capacity,ritter2005quantum}. 
    
    To create the model using Werner states and depolarizing channels, we follow the entanglement distribution protocol introduced in Section \ref{sect:III}. This protocol can be rewritten in terms of the processes that can lead to fidelity reduction during ARC-R operation, as follows.
    \begin{enumerate}
    \item Noisy EPR-pair (Werner state) generation in EPPSs and subsequent entanglement distribution within an elementary link, including storage inside the QMs and FFSMM operation.
    \item Entanglement distribution to the edge QMs in an ARC by means of intermediate BSM operations between elementary links.
    \item Quantum state transfer into the QR's internal memory involving all the processes detailed in Fig. \ref{fig:quantStateTransfer} and explained throughout.
    \item Storage within the QR's internal memory for a predefined time period $\tau$, determined according to \textbf{Lemma \ref{lemma:III}} for a specific network design.
    \item QR-QR entanglement swapping operations leading to entanglement distribution to the end nodes of the ARC-R.
    \end{enumerate}
    This stratification of the operations allows a bottom up approach to the fidelity analysis. Furthermore, we consider, here, only the buffered-router-assisted automated repeater chain as proposed and discussed in Sections \ref{sect:II} and \ref{sect:III}. Near-term fidelity parameters for each operation are extracted directly from recent literature results, which are detailed in Table \ref{tab:paramsFid}, along with long-term and ideal parameters, analogously to those efficiency parameters previously detailed in Table \ref{tab:params1}.
    
    \begin{table}[ht]
    \centering
    \renewcommand{\arraystretch}{1.1}
    \begin{tabular}{c | c || c | c}
    Symbol & Description & Near-Term & Long-Term\\
    \hline
    \hline
    \multicolumn{4}{c}{Automated Repeater Chain}\\
    \hline
    $\mathcal{F}^{\text{EPPS}}$ &
    EPPS Fidelity w.r.t. $\ket{\Phi^{+}}$ &
    93.3\% \cite{liu2021experimental} &
    99.0\%\\
    $\mathcal{F}^{\text{AFC}}$ &
    AFC QM storage Fidelity &
    96.8\% \cite{zhong2017nanophotonic} &
    99.0\%\\
    $\mathcal{F}^{\text{BSM}}$ &
    LO BSM Proj. Fidelity &
    97.2\% \cite{zhang2019experimental} &
    99.0\%\\
    $\mathcal{F}^{\text{FFSMM}}$ &
    FFSMM map. Fidelity &
    97.0\% \cite{sinclair2014FFSMM} &
    99.0\%\\
    \hline
    \multicolumn{4}{c}{Quantum State Transfer}\\
    \hline
    $\mathcal{F}^{\text{BUFF}}$ &
    Buffer storage Fidelity &
    99.6\% \cite{cho2016highly} &
    99.9\%\\
    $\mathcal{F}^{\text{QFC}}$ &
    QFC conv. Fidelity &
    99.8\% \cite{bock2018high} &
    99.9\%\\
    $\mathcal{F}^{\text{TB-POL}}$ &
    T.B.Q.$\rightarrow$Pol.Q. conv. Fidelity &
    97.4\% \cite{Pan-PolarizBinToTimebin}*&
    99.0\%\\
    $\mathcal{F}^{\text{MAP}}$ &
    Phot.Q.$\rightarrow e^{-}$ map. Fidelity &
    94.4\% \cite{bhaskar2020experimental}* &
    99.0\%\\
    \hline
    \multicolumn{4}{c}{Quantum Router}\\
    \hline
    $\mathcal{F}^{^{13}\text{C}}$ &
    $e^{-}\!\rightarrow\!^{13}\!C$ swap Fidelity &
    99.7\% \cite{bradley2019ten} &
    99.9\%\\
    $\mathcal{F}^{\text{CNOT}}$ &
    $^{13}\!C-^{13}\!C$ CNOT Fidelity &
    97.2\% \cite{bradley2019ten} &
    99.0\%\\
    $\mathcal{F}^{\text{R-OUT}}$ &
    Readout Fidelity &
    94.5\% \cite{bradley2019ten} &
    99.0\%\\
    \hline
    \end{tabular}
    \caption{Detailed fidelity parameters of all platforms utilized in the multi-platform network. *Estimated based on experimental parameters provided by authors.}
    \label{tab:paramsFid}
    \end{table} 
    
    Starting from the generation of the noisy EPR-pairs in the EPPSs, the corresponding density matrix, written as a Werner state, is:
    \begin{equation}
        \rho^{\text{EPPS}} = W^{\text{EPPS}}|\phi^+\rangle \langle \phi^+| + \frac{1-W^{\text{EPPS}}}{4}~\mathbb{I}_4,
        \label{eq:densMat_epps}
    \end{equation}
    where $W^{\text{EPPS}}$ is the Werner parameter calculated from $\mathcal{F}^{\text{EPPS}}$ as $W^{\text{EPPS}} = \tfrac{4\mathcal{F}^{\text{EPPS}}-1}{3}$, $|\phi^+\rangle = \frac{1}{\sqrt{2}}(|00\rangle + |11\rangle)$, and $\mathbb{I}_4$ is the identity matrix of dimension $4\times4$. In the following, the structure of the density matrix presented in Eq. \ref{eq:densMat_epps} will appear frequently, as will the Werner parameter and its relationship with the fidelity parameter. In order to simplify the next few steps, we establish the following notation for the density matrix $\rho'_W$ of an arbitrary Werner state with Werner parameter $W'$, and associated fidelity with respect to an ideal EPR-pair $\mathcal{F}'$:
    \begin{equation}
        \begin{split}
            \rho'_W\left(W'\right) &:= W'|\phi^+\rangle \langle \phi^+| + \frac{1-W'}{4}~\mathbb{I}_4;\\
            W' &:= \tfrac{4\mathcal{F}'-1}{3};\\
            \mathcal{F}' &:= \tfrac{3W'+1}{4}.
        \end{split}\label{eq:fidNotation}
    \end{equation}
    
    Within an elementary link, following a BSM remote operation, entanglement is swapped between those states previously stored in the quantum memories, which are retrieved and undergo a feed-forward spectral mode-mapping operation. By modeling the BSM operation, the storage-and-retrieval fidelity in the QM, and the subsequent FFSMM operation, each as a depolarizing channel, one can write the density matrix at the output of the processes $\rho_{\text{out}}\left(W_{\text{out}}\right)$, given the density matrix at their inputs $\rho_{\text{in}}\left(W_{\text{in}}\right)$, as:
    \begin{equation}
        \rho_{\text{out}}\left(W_{\text{out}}\right) = \alpha\rho_{\text{in}}\left(W_{\text{in}}\right) + \left(1-\alpha\right)\frac{\mathbb{I}}{4},
    \end{equation}
    where $\alpha$ can be thought of as the transition probability of these binomial processes and is intimately tied to their fidelity. In fact, this relationship is the same as the one presented in Eq. \ref{eq:fidNotation}. Therefore, by replacing $\alpha$ by $W$, for each of the processes considered above, one finds the density matrix associated to the imperfect EPR-pair distributed within an ARC's elementary link to be $\rho^{\text{elem}}_W\left(W^{\text{elem}}\right)$, with
    \begin{equation}
        W^{\text{elem}} = W^{\text{BSM}}\left(W^{\text{EPPS}}W^{\text{AFC}}W^{\text{FFSMM}}\right)^2.
        \label{eq:fid_elemLink}
    \end{equation}
    The expression in Eq. \ref{eq:fid_elemLink}, in analogy with Eq. \ref{eq:P_elemLink}, encapsulates the fidelity of the noisy EPR-pair distributed within an elementary link of an ARC. An ARC consisting of n such elementary links requires $n$ steps where  $\rho_W^{\text{elem}}$ is generated within the elementary links and $n-1$ intermediate BSM operations. Applying the same procedure as previously, this leads to a density matrix $\rho_W^{\text{ARC}}\left(W^{\text{ARC}}\left(n\right)\right)$ with an associated Werner parameter:
    \begin{equation}
        W^{\text{ARC}}\left(n\right) = \left(W^{\text{BSM}}\right)^{\left(n-1\right)}\left(W^{\text{elem}}\right)^n.
    \end{equation}
    
    After entanglement distribution to the QMs at the edges of an ARC of length $n$, the quantum state transfer into the QR's internal memory follows. The overall impact on the fidelity of the quantum state can be agglutinated in a parameter $W^{\text{QST}}$ that describes the effect of the processes, each considered as a depolarizing channel:
    \begin{equation}
        W^{\text{QST}} = \left(W^{\text{BUFF}}W^{\text{QFC}}W^{\text{TB-POL}}W^{\text{MAP}}\right).
        \label{eq:w_qst}
    \end{equation}
    A comment must be made at this point, regarding network configurations A and B. In the former case, an extra $(W^{AFC})^{\left(n-1\right)}$ must be included in the above expression to account for the extra $\left(n-1\right)$ QMs included at the edges of the ARCs. Also regarding Eq. \ref{eq:w_qst}, the effect of the buffer is simplified due to the fact that the longer possible storage in the buffer (500 $\mu$s) is much shorter than the spin coherence time (1 ms). The density matrix of the state stored within the internal memories of two adjacent QRs interconnected through an ARC of length $n$ is, thus, $\rho^{\text{QR-QR}}_W\left(W^{\text{QR-QR}}\left(n\right)\right)$, with an associated Werner parameter:
    \begin{equation}
        W^{\text{QR-QR}}\left(n\right) = W^{\text{ARC}}\left(n\right)\left(W^{\text{QST}}\right)^2.
        \label{eq:fid_arc}
    \end{equation}
    
    We have, so far, covered points (i) to (iii) of the ARC-R operation, the ones that involve entanglement distribution between two adjacent QRs mediated by an ARC of length $n$; this is in analogy with \textbf{Lemma \ref{lemma:I}}, which determines the rate in such a scenario. In the next step, the QRs store the quantum states inside their internal memories for up to a time $\tau$, as defined in \textbf{Lemma \ref{lemma:III}}. We consider a worst-case scenario, where the states are always stored for the longer amount of time and, thus, experience maximum decoherence during storage. This decoherence, modeled as a depolarizing channel, assumes the form
    \begin{equation}
        \rho_{\tau} := e^{\left(-b\tau\right)} \rho_W + \left(1- e^{\left(-b\tau\right)}\right)\frac{\mathbb{I}}{4},
    \end{equation}
    for a given initial density matrix $\rho_W$, where the exponential parameter is $b=1/3$, as per \cite{rozpedek2019building}. The storage occurs after a swap operation from the electron spin to the internal memory of the QR, such that the density matrix that represents the states shared by two adjacent QRs, after a storage time $\tau$, 
    is $\rho^{\text{QR-QR}, \tau}_W\left(W^{\text{QR-QR}}\left(n,\tau\right)\right)$, with an associated Werner parameter:
    \begin{equation}
        W^{\text{QR-QR}}\left(n,\tau\right) = W^{\text{QR-QR}}\left(n\right)\left(W^{^{13}\text{C}}e^{\left(-b\tau\right)}\right)^2.
    \end{equation}
    
    At the last stage of the entanglement distribution process, CNOT-based BSM operations followed by the state readout are performed within the intermediate routers. For an ARC-R of length $N$, i.e., with $N$ ARCs interconnecting $N-1$ QRs, we are left with a density matrix $\rho^{\text{ARC-R}}_W\left(W^{\text{ARC-R}}\left(n,N,\tau\right)\right)$, with associated Werner parameter
    \begin{equation}
        \begin{split}
        W^{\text{ARC-R}}\left(n,\tau,N\right) &= \left(W^{\text{QR-QR}}\left(n,\tau\right)\right)^N\times\\ &\hspace{-1.5cm}\left(\left(W^{^{13}\text{C}}e^{\left(-b\tau\right)}\right)^2W^{\text{CNOT}}W^{\text{R-OUT}}\right)^{N-1},
        \label{eq:fid_arc-r}
        \end{split}
    \end{equation}
    and, finally, an estimated fidelity with respect to an ideal maximally-entangled Bell state
    \begin{equation}
        \mathcal{F}^{\text{ARC-R}}\left(n,\tau,N\right) = \frac{3W^{\text{ARC-R}}\left(n,\tau,N\right)+1}{4}. 
    \end{equation}
    
    Eqs. \ref{eq:densMat_epps}-\ref{eq:fid_arc-r} provide a simplified means to analyze the fidelity of the states distributed at different steps---and with different parameters---of the proposed network. In order to provide an analysis of the achievable fidelity results in conjunction with those of the rates, we focus on Eqs. \ref{eq:fid_arc} and \ref{eq:fid_arc-r} to present Fig. \ref{fig:fidelity_near_long}, where, for simplicity, only configuration A is analyzed, since the fidelity results for configuration B are dramatically worse.
    
    \begin{figure}[ht]
    	\centering
    	\includegraphics[width=0.8\linewidth]{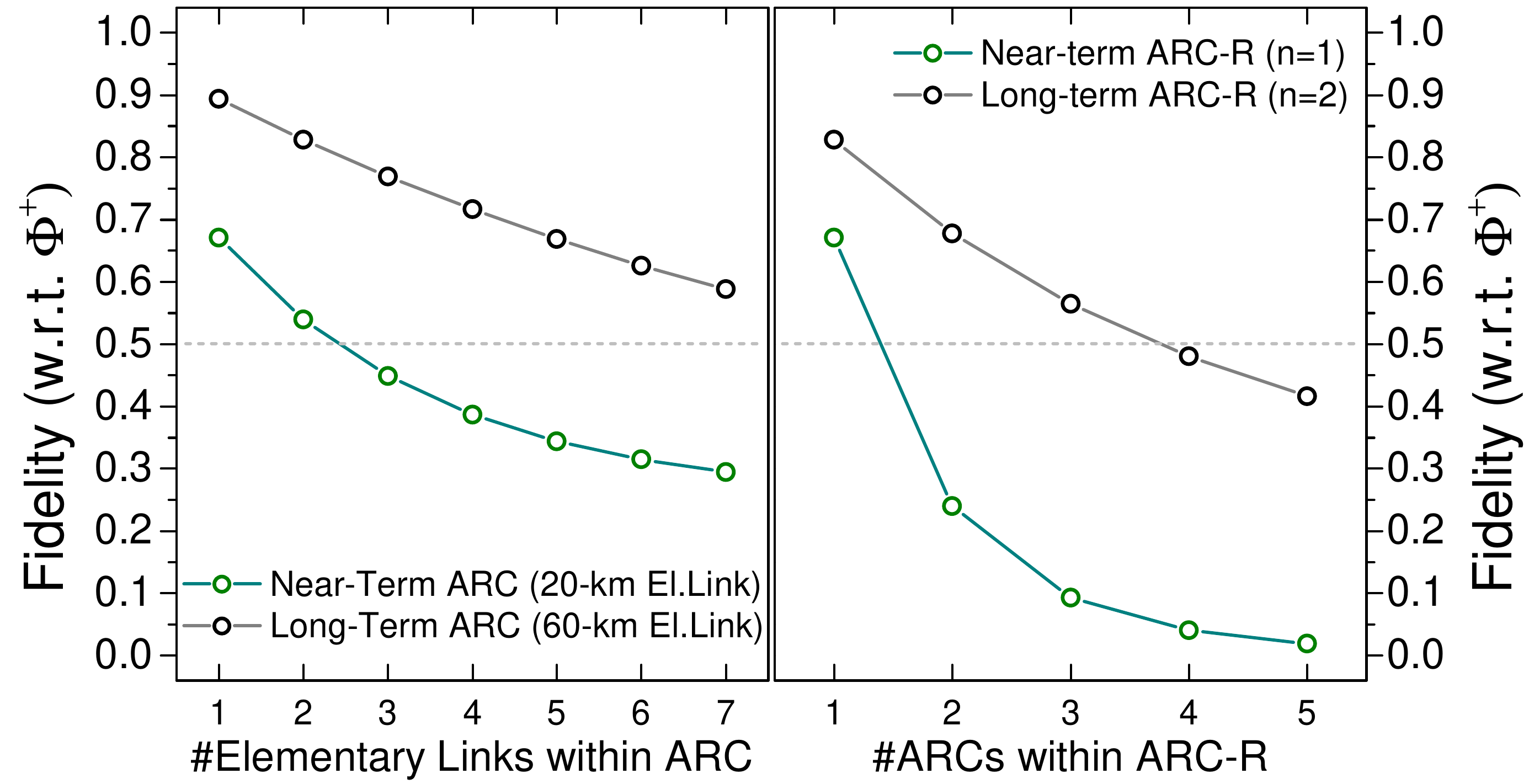}
    	\caption{Achievable fidelity values (with respect to a maximally-entangled Bell state) of the states distributed by the proposed multi-platform network. On the left pane, the fidelity is depicted as a function of the number of elementary links within an ARC; this corresponds to the QR-QR fidelity, or Eq. \ref{eq:fid_arc}. On the right panel, the fidelity is estimated based on Eq. \ref{eq:fid_arc-r} and is depicted as a function of the length of the ARC-R in terms of how many ARCs compose it.}
    	\label{fig:fidelity_near_long}
    \end{figure}
    
    We start by outlining, on the left panel of Fig. \ref{fig:fidelity_near_long}, the performance of a single ARC that connects two adjacent QRs in terms of the achievable fidelity---with respect to a maximally-entangled Bell-sate---of the distributed states as the number of elementary links that compose the ARC increases for near- and long-term parameters. Here, the elementary link lengths are kept fixed at 20 and 60km, respectively, in accordance with previous rate results analysis. The results show that the fidelity rapidly decreases for near-term parameters, though it is possible to distribute states with a fidelity above the threshold of $50\%$, which would still allow for entanglement distillation to take place, for short ARCs, i.e., QRs spaced by 20 km. This is a promising result as far as near-term goals are concerned, especially when analyzed in conjunction with the results of Fig. \ref{fig:NV vs AFC}, i.e., that a rate of roughly one noisy EPR-pair per second can be achieved with a fidelity of $\sim70\%$.
    
    On the right panel of Fig. \ref{fig:fidelity_near_long}, the fidelity values for an ARC-R are depicted, for an increasing number of constituent ARCs. Here, the number of elementary links that compose a single ARC are kept at $n=1$ and $n=2$ for near- and long-term parameters, respectively. As expected, the results of near-term parameters do not allow for distribution of useful states across a network longer than 20km, since the fidelity quickly drops below the $50\%$ threshold. For long-term parameters, however, it would be possible to operate, for a nearly 200km-long ARC-R, with a fidelity of $\sim60\%$ and a distribution rate of 100 noisy EPR-pairs per second. Moreover, we point out that, in the long-term regime, a $120$km-long ARC-R achieves an end-to-end fidelity ($\mathcal{F}^{\text{ARC-R}}\left(n,\tau,N\right)$) well above $\sim80\%$. The estimated quantum bit-error rate (QBER) for such noisy EPR-pairs (i.e., the probability of observing two different measurement outcomes of the shared EPR-pair) is given by $\frac{2}{3}(1-\mathcal{F}^{\text{ARC-R}}\left(n,\tau,N\right)) = \frac{2}{3}(1-0.8)= 0.1333$. This allow the extraction of secret keys using sophisticated quantum key distribution protocols like the so-called two-way six-state protocol \cite{WMU07}. 
    
    \section{Discussion}
    
    The analyses conducted in Sections \ref{sect:IV} and \ref{sect:V}, even assuming simplified models, allow bringing the concept of such a multi-platform quantum network to a more tangible realm. Beyond the hard values of rates and fidelity estimated for the different network topologies discussed here, which are intrinsically tied to the parameters compiled in Tables \ref{tab:params1} and \ref{tab:paramsFid}, the potential of combining distinct platforms that operate under different protocols, wavelengths, and quantum information encoding should be highlighted. The possibility of harnessing specific resources from each of these platforms by assigning them to individual roles in the network design is a natural trend in the future of quantum communications.
    
    Despite the fact that the natural evolution of the platforms considered here (and also of those ones not considered) will impact the viability of the proposed network, in both its constituents and design, useful observations can be made from its current form. Probably the most striking of those is the balance between high entanglement distribution rates and fidelity of the distributed states with respect to an ideally maximally-entangled state. Although the objective is to be able to reach high values for both, there might be certain regimes where the rate could be augmented in detriment of the fidelity -- the distinction between configurations A and B discussed in Section \ref{sect:III} is an example of such regime. This strategy, if in keeping with the distributed state's fidelity threshold of $50\%$, can benefit from the technique of entanglement distillation \cite{BBPS96} inside the QRs (which have the potential to implement such techniques \cite{kalb2017entanglement}), or even at the end-nodes, to establish high quality entanglement at a reasonable rate.
    
    Throughout the analysis of the network, several simplifications were made, which, if revisited, could provide an avenue for optimizations of the figures-of-merit currently estimated. These are, for instance, the fact that the length of the ARC was considered to be the same within an ARC-R. Although this has a significant positive impact on the complexity of the synchronization of the network, it corresponds to a particular solution of the problem, and analyzing different possible configurations is an interesting optimization problem \cite{rabbie2020designing}. Besides the design optimizations, it is very clear that the entanglement distribution protocol can also be optimized, also associated to an increase of complexity of network synchronization. Specifically, performing the CNOT-based BSM inside the QRs based on availability of states distributed by an ARC and not on a pre-determined parameter, such as the one introduced here ($\tau$). 
    
    An important remark is that the most compelling argument for the inclusion of the buffers in the network design, the possibility of connecting more than two ARCs to a single QR, has not been analyzed deeply. The reason behind it is the drastic increase in complexity involved in analyzing the behavior of a point-to-multipoint, or even multipoint-to-multipoint, network compared to analyzing a simple point-to-point chain network. Even though the performance boost with respect to an equivalent network without buffers in the former cases is expected to be more impactful than those presented in Fig. \ref{fig:bufferComparison} for a chain network, it is left as a future point of investigation, along with the optimization possibilities introduced above. Finally, it is important to mention that, although SPDC-based EPPSs have been introduced in the design due to their potential of creating spectrally-multiplexed photon-pairs at high rates, it is still the platform that provides the most limitations, both in terms of efficiency and fidelity; this effect is magnified due to the high number of such sources required to build the proposed network. In the future, a promising technology such as semiconductor quantum dots \cite{laferriere2020multiplexed} could be a substitute for SPDCs in what we conjecture to be the bottleneck of the network, the EPPSs.
   
    \section{Conclusion}
    
    A multi-platform quantum network is proposed and analysed in terms of achievable entanglement distribution rates and expected fidelity of the distributed states with respect to a maximally-entangled state, where we capitalize on unique features of its individual constituents. This is possible due to a stratification of the entanglement distribution protocol within three main structures: the elementary links that compose a quantum repeater chain; the so-called automated repeater chains; and the buffered-router-assisted automated repeater chain. We demonstrate that, even though the scaling of the rates, with an increasing network length, is harsher for the proposed network than for its NV-based counterpart, the boost due to frequency-multiplexing creates regimes where it becomes more profitable. By compiling current state-of-the-art parameters of efficiencies and fidelity, we were able to contextualize the network and demonstrate its advantages, especially when long-term parameters (based on the current level of maturity of most of these technologies) are considered. The current proposal not only introduces a novel network topology design and entanglement distribution protocol, but also opens up new avenues of optimizations of the design for future works.
    
    \section*{Acknowledgements}
    The project was developed during the authors' post-doctorate appointments at QuTech, where the lively atmosphere of discussion and contribution motivated this work. KC thanks Stephanie Wehner, Filip Rozpedek, and Guus Avis for useful discussions. 

    \section*{References}


\providecommand{\newblock}{}

\end{document}